\documentclass[twocolumn,pra]{revtex4-1}%
\usepackage{amsfonts}
\usepackage{hyperref}
\usepackage{amsmath}
\usepackage{amssymb}
\usepackage{graphicx}%
\providecommand{\U}[1]{\protect\rule{.1in}{.1in}}
\newtheorem{theorem}{Theorem}

\newtheorem{corollary}{Corollary}

\newtheorem{definition}{Definition}

\newtheorem{remark}{Remark}

\newenvironment{proof}[1][Proof]{\noindent\textbf{#1.} }{\ \rule{0.5em}{0.5em}}
\begin{document}
\preprint{ }
\title[ ]{Deconstruction and conditional erasure of quantum correlations}
\author{Mario Berta}
\affiliation{Department of Computing, Imperial College London, London SW7 2AZ, UK}
\author{Fernando G.~S.~L.~Brand\~{a}o}
\affiliation{Institute for Quantum Information and Matter, California Institute of
Technology, Pasadena, California 91125, USA}
\author{Christian Majenz}
\affiliation{Institute for Language, Logic and Computation, University of Amsterdam, and QuSoft, 1098XG Amsterdam, Netherlands}
\affiliation{Department of Mathematical Sciences, University of Copenhagen,
Universitetsparken 5, DK-2100 Copenhagen \O}
\author{Mark M.~Wilde}
\affiliation{Hearne Institute for Theoretical Physics, Department of Physics and Astronomy,
Center for Computation and Technology, Louisiana State University, Baton
Rouge, Louisiana 70803, USA}

\begin{abstract}
We define the deconstruction cost of a tripartite quantum state on systems
$ABE$ as the minimum rate of noise needed to apply to the $AE$ systems, such
that there is negligible disturbance to the marginal state on the $BE$ systems,
while the system $A$ of the resulting state is locally recoverable from the $E$
system alone. We refer to such actions as deconstruction operations and
protocols implementing them as state deconstruction protocols. State
deconstruction generalizes Landauer erasure of a single-party quantum state as
well the erasure of correlations of a two-party quantum state. We find that
the deconstruction cost of a tripartite quantum state on systems $ABE$ is
equal to its conditional quantum mutual information (CQMI) $I(A;B|E)$, thus
giving the CQMI an operational interpretation in terms of a state
deconstruction protocol. We also define a related task called conditional erasure,
in which the goal is to apply noise to systems $AE$ in order to decouple
system $A$ from systems $BE$, while causing negligible disturbance to the
marginal state of systems $BE$. We find that the optimal rate of noise for
conditional erasure is also equal to the CQMI $I(A;B|E)$. State deconstruction
and conditional erasure lead to operational interpretations of the quantum
discord and squashed entanglement, which are quantum correlation measures
based on the CQMI. We find that the quantum discord is equal to the cost of
simulating einselection, the process by which a quantum system interacts with
an environment, resulting in selective loss of information in the system. The
squashed entanglement is equal to half the minimum rate of noise needed for
deconstruction/conditional erasure if Alice has available the best possible
system $E$ to help in the deconstruction/conditional erasure task.

\end{abstract}
\volumeyear{ }
\volumenumber{ }
\issuenumber{ }
\eid{ }
\date{\today}
\startpage{1}
\endpage{10}
\maketitle

\section{Introduction}

The Landauer erasure principle represents a deep link between information
theory and thermodynamics \cite{L61}. An informal summary of the principle is
that the work cost of erasing the contents of a computer memory is
proportional to the amount of information stored there. This insight has now
sparked a whole literature, a consequence of which has been a deepening of the
connection between information theory and thermodynamics (see, e.g.,
\cite{GHRRS15} for a review).

One generalization of Landauer's insight goes beyond the single-system setup
mentioned above. In \cite{GPW05}, Groisman \textit{et al}.~considered a
setting in which two parties share a quantum state $\rho_{AB}$. Their goal was
to determine the work cost of erasing the correlations present in the state,
by acting locally on one system, such that the resulting state has a
tensor-product form $\sigma_{A}\otimes\omega_{B}$, where $\sigma_{A}$ and
$\omega_{B}$ are quantum states. Groisman \textit{et al}.~solved the problem
in the framework of quantum Shannon theory \cite{W16}, whereby they allowed
the two parties to have many copies of the state $\rho_{AB}$ and quantified
the minimum rate of noise that needs to be applied to the $A$~systems such
that the resulting state is tensor-product between the $A$~systems and the
$B$~systems. They found that the optimal rate of noise is equal to the quantum
mutual information of the state $\rho_{AB}$, defined as%
\begin{equation}
I(A;B)_{\rho}\equiv H(A)_{\rho}+H(B)_{\rho}-H(AB)_{\rho},
\end{equation}
where the quantum entropy of a state $\sigma_{G}$ on system $G$ is defined as
$H(G)_{\sigma}\equiv-\operatorname{Tr}\{\sigma_{G}\log_{2}\sigma_{G}\}$. An
important consequence of their theorem is that we can assign a physical
meaning to, or operational interpretation of, the quantum mutual information
as the minimum rate of noise needed to completely erase the correlations
present in a two-party quantum state. Thus, we can say that quantum mutual
information is equal to the work cost of correlation destruction.

On the other hand, quantum mutual information has also been interpreted in a
communication-theoretic task (now called \textit{coherent state merging}
\cite{O08a}) as the optimal rate of entanglement creation when transferring
the system $A$ of $\rho_{AB}$ to a party possessing system $B$
\cite{ADHW06FQSW}, while using quantum communication at a fixed rate. These
dual interpretations of quantum mutual information in terms of destruction and
creation perhaps come at no surprise if one is familiar with the unitarity of
quantum mechanics and the purification principle. Information can never truly
be destroyed in quantum mechanics, which means that the apparent destruction
of correlations between two parties implies the creation of correlations
elsewhere, i.e., with another party who possesses a purification of the state
$\rho_{AB}$. In fact, this insight is the main idea underlying the decoupling
principle \cite{qip2002schu,qcap2008first}, which is a method for proving the
above theorem \cite{ADHW06FQSW}\ and others similar to it.

In this paper, we are interested in further generalizations of the erasure of
correlations to a three-system scenario, i.e., for a tripartite quantum state
$\rho_{ABE}$
 (see also our companion paper \cite{BBMW18}).
 The tasks we are
interested in accomplishing are more delicate than the destruction of
correlations mentioned above.

The first task we consider is a \textit{state deconstruction} protocol, whose
aim is to deconstruct (literally, \textquotedblleft to break into constituent
components\textquotedblright) the correlations in a three-party quantum state.
To make the setting precise, consider a state $\rho_{ABE}$, and suppose that
Alice possesses system $A$, Bob system $B$, and Eve system $E$. We would like
a deconstruction protocol to result in a state for which Eve is the mediator
of correlations between Alice and Bob, while the original correlations shared
between Eve and Bob are negligibly disturbed. The setup begins with Alice and
Eve in the same laboratory and Bob in a different laboratory, and we also
operate in the framework of quantum Shannon theory, allowing them to share $n$
copies of the state $\rho_{ABE}$, where $n$ can be a large number. Following
Groisman \textit{et al}.~\cite{GPW05}, we allow for a local unitary
randomizing channel acting on the $AE$ systems and an ancilla. The rate of
noise is equal to the logarithm of the number of unitaries in such a channel
divided by the number $n$ of copies of the state $\rho_{ABE}$. We define the
\textit{deconstruction cost} of a tripartite state $\rho_{ABE}$\ to be the
minimum rate of noise needed to apply to the $AE$ systems and an ancilla, such
that the resulting state satisfies the following:

\begin{enumerate}
\item the resulting system of Alice is \textit{locally recoverable} from Eve's
system alone, and

\item the correlations between Eve and Bob are \textit{negligibly disturbed}.
\end{enumerate}

\noindent See Section~\ref{sec:dec-protocol}\ for a more detailed definition
and Figure~\ref{fig:deconstruction}\ for a depiction of a state deconstruction
protocol along with the conditions of local recoverability and negligible disturbance.

The second task we consider is \textit{conditional erasure}. Such a task is
very similar to state deconstruction:\ we allow for a local channel to act on
the $AE$ systems and an ancilla. However, we define the conditional erasure
cost to be the minimum rate of noise such that the resulting system of Alice
is decoupled from the $BE$ systems and the marginal state of the $BE$ systems
is negligibly disturbed. A protocol that accomplishes conditional erasure also
accomplishes state deconstruction:\ this is because a decoupled system is
locally recoverable.

The negligible disturbance condition is critical in both state deconstruction
and conditional erasure: it could be the case that Eve and Bob would want to
use their systems for some later quantum information processing task, so that
keeping the correlations intact is essential for the systems to be useful
later on. For example, Eve's and Bob's systems might contain some entanglement
which could be useful for a subsequent distributed quantum computation. This
condition also highlights an essential difference between semi-classical and
fully quantum protocols: in the case that the system $E$ is classical, the
negligible disturbance condition is not necessary because one could always
observe the value in Eve's system without causing any disturbance to it.
However, in the quantum case, the uncertainty principle forbids us from taking
a similar action, so that it is necessary for fully quantum protocols to
proceed with a greater sleight of hand.

State deconstruction and conditional erasure are far more delicate than
decoupling, the latter sometimes described as having the \textquotedblleft
relatively indiscriminate goal of destruction\textquotedblright%
\ \cite{ADHW06FQSW}. That is, a naive application of the decoupling method is
too blunt of a tool to apply in these protocols. Applying it naively would
result in the annihilation of correlations such that if correlations between
systems $B$ and $E$ were present beforehand, they would be destroyed and thus
no longer useful for a future quantum information processing task.

\section{Main result}

\label{sec:main-result}The main result of this paper is that both the
deconstruction cost and the conditional erasure cost of a tripartite state
$\rho_{ABE}$ are equal to its conditional quantum mutual information (CQMI),
defined as%
\begin{equation}
I(A;B|E)_{\rho}\equiv I(AE;B)_{\rho}-I(E;B)_{\rho}.
\end{equation}
(See Theorems~\ref{thm:main} and \ref{thm:cond-erasure-region}.) Thus, our
result assigns a new physical meaning to the CQMI, in terms of erasure or
thermodynamical tasks that generalize Landauer's original scenario as well as
the erasure of correlations scenario from \cite{GPW05}. The deconstruction and
conditional erasure tasks are intimately related to properties of the
CQMI\ itself, which has previously been related to local recoverability
\cite{HJPW04,FR14,BSW14}\ as well as the condition of negligible disturbance
\cite{WDHW13}.

The state deconstruction and conditional erasure tasks are also closely
related to the protocol of quantum state redistribution \cite{DY08,YD09},
which, prior to our contribution, was the only protocol giving an operational
meaning for the CQMI. A quantum state redistribution protocol begins with many
independent copies of a four-party pure state $\psi_{ABER}$, with a sender
possessing the $A$ and $E$ systems, a receiver possessing the $R$ systems, and
the sender and receiver sharing noiseless entanglement before communication
begins. The main result of \cite{DY08,YD09} is that the optimal rate of
quantum communication needed to redistribute the $A$~systems from the sender
to the receiver is equal to $\tfrac{1}{2}I(A;B|R)_{\psi}$. In the present
paper, the state redistribution protocol is one of the main tools that we use
for establishing that the deconstruction and conditional erasure costs are
each equal to the CQMI.

The other main tool that we use is a quantity known as the \textit{fidelity of
recovery} of a tripartite state $\rho_{ABE}$ \cite{SW14}:%
\begin{equation}
F(A;B|E)_{\rho}\equiv\sup_{\mathcal{R}_{E\rightarrow AE}}F(\rho_{ABE}%
,\mathcal{R}_{E\rightarrow AE}(\rho_{BE})),
\end{equation}
where the quantum fidelity between states $\omega$ and $\tau$ is defined as
$F(\omega,\tau)\equiv\left\Vert \sqrt{\omega}\sqrt{\tau}\right\Vert _{1}^{2}$
\cite{U73} and the supremum is with respect to all recovery channels
$\mathcal{R}_{E\rightarrow AE}$.

Our main results then lead to operational interpretations of quantum
correlation measures based on CQMI, including quantum discord
\cite{Z00,zurek01} and squashed entanglement \cite{CW04}. We find that the
quantum discord is equal to the optimal rate of simulating einselection
\cite{Z03}, the process by which a system interacts with an environment in
such a way as to cause selective loss of information in the system. In
particular, given a bipartite state $\rho_{AB}$ and measurement $\Lambda_{A}$,
we find that the discord is equal to the minimum rate of noise needed to apply
to the $A$ system of $\rho_{AB}$, such that the resulting state is locally
recoverable after performing a measurement on the $A$ system and its
post-measurement state is indistinguishable from the post-measurement state
after $\Lambda_{A}$ acts on $\rho_{AB}$. We find that the squashed
entanglement of a state $\rho_{AB}$ is equal to half the minimum rate of noise
needed in a deconstruction operation which has the best possible quantum side
information in system $E$ to help in the deconstruction task.

An outline of the rest of the paper is as follows. In
Section~\ref{sec:background}, we provide more background on quantum
information basics and the conditional quantum mutual information, and we
review the state redistribution protocol in more detail.
Section~\ref{sec:dec-protocol}\ defines a state deconstruction protocol and
the deconstruction cost of a tripartite state $\rho_{ABE}$, and
Section~\ref{sec:landauer-erasure-model} discusses a slightly different model
for state deconstruction. In Section~\ref{sec:CQMI-lower-bound}, we prove that
the deconstruction cost is bounded from below by the CQMI. After that,
Section~\ref{sec:simulation-via-redis}\ proves the other inequality, by
showing how a state redistribution protocol leads to one for state
deconstruction. In Section~\ref{sec:cond-erasure}, we define the conditional
erasure task and show how a conditional erasure protocol is equivalent to a
quantum state redistribution protocol, in the sense that the existence of one
implies the existence of the other. We then establish the CQMI\ as the optimal
conditional erasure cost. Section~\ref{sec:einselection-cost-discord}\ details
how quantum discord is equal to the optimal rate of einselection simulation,
and the following section gives the aforementioned operational interpretation
of squashed entanglement. We finally conclude in Section~\ref{sec:conclu}%
\ with a summary and some open questions.

\section{Background\label{sec:background}}

\subsection{Basics of quantum information\label{sec:basics}}

We review some basic aspects of quantum information before proceeding with the
main development (see, e.g., \cite{W16}\ for a review). Let $\mathcal{L}%
(\mathcal{H})$ denote the algebra of bounded linear operators acting on a
Hilbert space $\mathcal{H}$ (we consider finite-dimensional Hilbert spaces
throughout this paper). Let $\mathcal{L}_{+}(\mathcal{H})$ denote the subset
of positive semi-definite operators. An\ operator $\rho$ is in the set
$\mathcal{D}(\mathcal{H})$\ of density operators (or states) if $\rho
\in\mathcal{L}_{+}(\mathcal{H})$ and Tr$\left\{  \rho\right\}  =1$. Throughout
this paper, we let $\pi$ denote the maximally mixed state on a given Hilbert
space $\mathcal{H}$, so that $\pi\equiv I_{\mathcal{H}}/\dim(\mathcal{H)}$.
The tensor product of two Hilbert spaces $\mathcal{H}_{A}$ and $\mathcal{H}%
_{B}$ is denoted by $\mathcal{H}_{A}\otimes\mathcal{H}_{B}$ or $\mathcal{H}%
_{AB}$.\ Given a multipartite density operator $\rho_{AB}\in\mathcal{D}%
(\mathcal{H}_{A}\otimes\mathcal{H}_{B})$, we unambiguously write $\rho_{A}%
=\ $Tr$_{B}\left\{  \rho_{AB}\right\}  $ for the reduced density operator on
system $A$. We use $\rho_{AB}$, $\sigma_{AB}$, $\tau_{AB}$, $\omega_{AB}$,
etc.~to denote general density operators in $\mathcal{D}(\mathcal{H}%
_{A}\otimes\mathcal{H}_{B})$, while $\psi_{AB}$, $\varphi_{AB}$, $\phi_{AB}$,
etc.~denote rank-one density operators (pure states) in $\mathcal{D}%
(\mathcal{H}_{A}\otimes\mathcal{H}_{B})$ (with it implicit, clear from the
context, and the above convention implying that $\psi_{A}$, $\varphi_{A}$,
$\phi_{A}$ may be mixed if $\psi_{AB}$, $\varphi_{AB}$, $\phi_{AB}$ are pure).
A purification $|\phi^{\rho}\rangle_{RA}\in\mathcal{H}_{R}\otimes
\mathcal{H}_{A}$ of a state $\rho_{A}\in\mathcal{D}(\mathcal{H}_{A})$ is such
that $\rho_{A}=\operatorname{Tr}_{R}\{|\phi^{\rho}\rangle\langle\phi^{\rho
}|_{RA}\}$. An isometry $U:\mathcal{H}\rightarrow\mathcal{H}^{\prime}$ is a
linear map such that $U^{\dag}U=I_{\mathcal{H}}$. Often, an identity operator
is implicit if we do not write it explicitly (and should be clear from the context).

A linear map $\mathcal{N}_{A\rightarrow B}:\mathcal{L}(\mathcal{H}%
_{A})\rightarrow\mathcal{L}(\mathcal{H}_{B})$\ is positive if $\mathcal{N}%
_{A\rightarrow B}\left(  \sigma_{A}\right)  \in\mathcal{L}(\mathcal{H}%
_{B})_{+}$ whenever $\sigma_{A}\in\mathcal{L}(\mathcal{H}_{A})_{+}$. Let
id$_{A}$ denote the identity map acting on a system $A$. A linear map
$\mathcal{N}_{A\rightarrow B}$ is completely positive if the map
id$_{R}\otimes\mathcal{N}_{A\rightarrow B}$ is positive for a reference system
$R$ of arbitrary size. A linear map $\mathcal{N}_{A\rightarrow B}$ is
trace-preserving if $\operatorname{Tr}\left\{  \mathcal{N}_{A\rightarrow
B}\left(  \tau_{A}\right)  \right\}  =\operatorname{Tr}\left\{  \tau
_{A}\right\}  $ for all input operators $\tau_{A}\in\mathcal{L}(\mathcal{H}%
_{A})$. A quantum channel is a linear map which is completely positive and
trace-preserving (CPTP). A quantum channel $\mathcal{U}:\mathcal{L}%
(\mathcal{H}_{A})\rightarrow\mathcal{L}(\mathcal{H}_{B})$ is an isometric
channel if it has the action $\mathcal{U}(X_{A})=UX_{A}U^{\dag}$, where
$X_{A}\in\mathcal{L}(\mathcal{H}_{A})$ and $U:\mathcal{H}_{A}\rightarrow
\mathcal{H}_{B}$ is an isometry.

The trace distance between two quantum states $\rho,\sigma\in\mathcal{D}%
(\mathcal{H})$\ is equal to $\left\Vert \rho-\sigma\right\Vert _{1}$. It has a
direct operational interpretation in terms of the distinguishability of these
states. That is, if $\rho$ or $\sigma$ are prepared with equal probability and
the task is to distinguish them via some quantum measurement, then the optimal
success probability in doing so is equal to $\left(  1+\left\Vert \rho
-\sigma\right\Vert _{1}/2\right)  /2$. The trace distance and fidelity are
related by the Fuchs-van-de-Graaf inequalities \cite{FG98}:%
\begin{equation}
1-\sqrt{F(\rho,\sigma)}\leq\frac{1}{2}\left\Vert \rho-\sigma\right\Vert
_{1}\leq\sqrt{1-F(\rho,\sigma)}. \label{eq:F-v-d-G-ineq}%
\end{equation}
The rightmost quantity above is known to be a distance measure, satisfying the
triangle inequality, as proposed and shown in \cite{GLN05}. This quantity was
generalized to subnormalized states and given the name \textquotedblleft
purified distance\textquotedblright\ in \cite{TCR09}.

Let $\{|i\rangle_{A}\}$ denote the standard, orthonormal basis for a Hilbert
space $\mathcal{H}_{A}$, and let $\{|i\rangle_{B}\}$ be defined similarly for
$\mathcal{H}_{B}$. If the dimensions of these spaces are equal ($\dim
(\mathcal{H}_{A})=\dim(\mathcal{H}_{B})=d$), then we define the maximally
entangled state $|\Phi\rangle_{AB}\in\mathcal{H}_{A}\otimes\mathcal{H}_{B}$ as%
\begin{equation}
|\Phi\rangle_{AB}\equiv\frac{1}{\sqrt{d}}\sum_{i=1}^{d}|i\rangle_{A}%
\otimes|i\rangle_{B}.
\end{equation}
The generalized Pauli shift operator $X$\ is defined by $X_{A}|i\rangle
_{A}=|i\oplus1\rangle_{A}$, where addition is modulo $d$. The generalized
Pauli phase operator $Z$ is defined by $Z_{A}|k\rangle_{A}=\exp(2\pi
ik/d)|k\rangle_{A}$. The Heisenberg--Weyl group is defined as $\{X_{A}%
^{j}Z_{A}^{k}\}_{j,k\in\left\{  1,\ldots,d\right\}  }$, and satisfies%
\begin{equation}
\frac{1}{d}\operatorname{Tr}\{X_{A}^{j}Z_{A}^{k}\}=\delta_{d,j}\delta
_{d,k}.\label{eq:HW-trace}%
\end{equation}
The generalized Bell basis is defined as $\{|\Phi^{j,k}\rangle_{AB}%
\}_{j,k\in\left\{  1,\ldots,d\right\}  }$, where%
\begin{equation}
|\Phi^{j,k}\rangle_{AB}=(X_{A}^{j}Z_{A}^{k}\otimes I_{B})|\Phi\rangle_{AB}.
\end{equation}
It is an orthonormal basis as a consequence of \eqref{eq:HW-trace}.

\subsection{Conditional quantum mutual information}

Here we briefly provide more background on the conditional quantum mutual
information (CQMI). The CQMI is understood informally as quantifying the
correlations between systems $A$ and $B$ from the perspective of a party
possessing system $E$ \cite{DY08,YD09}. The CQMI is symmetric with respect to
the exchange of the $A$ and $B$ systems of a state $\rho_{ABE}$:
$I(A;B|E)_{\rho}=I(B;A|E)_{\rho}$. One of the powerful properties of the
CQMI\ is that it obeys a chain rule of the following form for a state
$\sigma_{A_{1}\cdots A_{n}BE}$:%
\begin{equation}
I(A_{1}\cdots A_{n};B|E)_{\sigma}=\sum_{i=1}^{n}I(A_{i};B|EA_{1}%
^{i-1})_{\sigma},
\end{equation}
where $A_{1}^{i-1}\equiv A_{1}\cdots A_{i-1}$, so that we can think of the
correlations between $A_{1}\cdots A_{n}$ and $B$, as observed by $E$, being
built up one system at a time. The CQMI is always non-negative $I(A;B|E)_{\rho
}\geq0$, an entropy inequality known as strong subadditivity
\cite{LR73,PhysRevLett.30.434}. A first relation of CQMI to recoverability was
established in \cite{HJPW04}, in which it was shown that $I(A;B|E)_{\rho}=0$
if and only if there exists a recovery quantum channel $\mathcal{R}%
_{E\rightarrow AE}$ such that the global state $\rho_{ABE}$\ can be
reconstructed by acting on one share $E$\ of the marginal state $\rho_{BE}$:%
\begin{equation}
\rho_{ABE}=\mathcal{R}_{E\rightarrow AE}(\rho_{BE}).
\end{equation}
More recently, it was shown that these results are robust \cite{FR14,BSW14}%
:\ the CQMI\ is approximately equal to zero (i.e., $I(A;B|E)_{\rho}\approx0$)
if and only if the global state is approximately recoverable by acting on one
share $E$ of the marginal $\rho_{BE}$ (i.e., $\rho_{ABE}\approx\mathcal{R}%
_{E\rightarrow AE}(\rho_{BE})$). In more detail, \cite{FR14} established the
inequality%
\begin{equation}
I(A;B|E)_{\rho}\geq-\log F(A;B|E)_{\rho},
\end{equation}
and \cite{FR14,BSW14} established a converse relation. Using some recent tools
\cite{Winter15} and the Fuchs-van-de-Graaf inequalities in
\eqref{eq:F-v-d-G-ineq}, the following refinement of the converse holds
\cite[Theorem~11.10.5]{W16}: if $F(A;B|E)_{\rho}\geq1-\varepsilon$ for
$\varepsilon\in(0,1)$, then%
\begin{multline}
I(A;B|E)_{\rho}\leq2\sqrt{\varepsilon}\log\left\vert B\right\vert
\label{eq:loc-rec-low-CQMI}\\
+\left(  1+\sqrt{\varepsilon}\right)  h_{2}(\sqrt{\varepsilon}/\left[
1+\sqrt{\varepsilon}\right]  ),
\end{multline}
where the binary entropy $h_{2}(x)$ is defined for $x\in(0,1)$ as%
\begin{equation}
h_{2}(x)\equiv-x\log_{2}x-\left(  1-x\right)  \log_{2}\left(  1-x\right)  ,
\end{equation}
with the property that $\lim_{x\rightarrow0}h_{2}(x)=0$. From the above, we
see that the CQMI\ is a witness to quantum Markovianity:\ if it is small, then
we can understand the correlations between $A$ and $B$ as being mediated by
system $E$ via the recovery channel $\mathcal{R}_{E\rightarrow AE}$.

\subsection{Quantum state redistribution}

\label{sec:QSR-review}This section provides some background on quantum state
redistribution \cite{DY08,YD09}. A quantum state redistribution protocol
begins with a sender, a receiver, and a reference party sharing many
independent copies of a four-system pure state $\psi_{ABER}$. The sender has
the $AE$ systems, the receiver the $R$ systems, and the reference the $B$
systems. The goal is to use entanglement and noiseless quantum communication
to redistribute the systems such that the sender ends up with the $E$ systems,
the receiver the $AR$ systems, and the reference the $B$ systems. As a side
benefit, the protocol can also generate entanglement shared between the sender
and receiver at the end.

More formally, let $n\in\mathbb{N}$, $M\in\mathbb{N}$, and $\varepsilon
\in\left[  0,1\right]  $. An $\left(  n,M,\varepsilon\right)  $ state
redistribution protocol consists of an encoding channel $\mathcal{E}%
_{A^{n}E^{n}A^{\prime}\rightarrow\bar{A}_{0}A_{0}\hat{E}^{n}}$\ and a decoding
channel $\mathcal{D}_{\bar{A}_{0}R^{\prime}R^{n}\rightarrow\hat{A}^{n}\hat
{R}^{n}R_{0}}$, such that the following state%
\begin{equation}
\xi_{\hat{A}^{n}B^{n}\hat{E}^{n}\hat{R}^{n}A_{0}R_{0}}\equiv\mathcal{D}%
_{\bar{A}_{0}R^{\prime}R^{n}\rightarrow\hat{A}^{n}\hat{R}^{n}R_{0}}%
(\varphi_{\bar{A}_{0}A_{0}\hat{E}^{n}B^{n}R^{n}R^{\prime}}),
\label{eq:final-state-redis-state}%
\end{equation}
where%
\begin{equation}
\varphi_{\bar{A}_{0}A_{0}\hat{E}^{n}B^{n}R^{n}R^{\prime}}\equiv\mathcal{E}%
_{A^{n}E^{n}A^{\prime}\rightarrow\bar{A}_{0}A_{0}\hat{E}^{n}}(\psi
_{ABER}^{\otimes n}\otimes\Phi_{A^{\prime}R^{\prime}}),
\label{eq:phi-state-redist}%
\end{equation}
has fidelity larger than $1-\varepsilon$ with the following pure state:%
\begin{equation}
\psi_{\hat{A}B\hat{E}\hat{R}}^{\otimes n}\otimes\Phi_{A_{0}R_{0}},
\end{equation}
where $\Phi_{A^{\prime}R^{\prime}}$ and $\Phi_{A_{0}R_{0}}$ denote maximally
entangled states of Schmidt ranks $|A^{\prime}|$ and $|A_{0}|$, respectively.
That is, an $\left(  n,M,\varepsilon\right)  $ state redistribution protocol
satisfies%
\begin{equation}
F(\xi_{\hat{A}^{n}B^{n}\hat{E}^{n}\hat{R}^{n}A_{0}R_{0}},\psi_{\hat{A}B\hat
{E}\hat{R}}^{\otimes n}\otimes\Phi_{A_{0}R_{0}})\geq1-\varepsilon.
\label{eq:good-QSR}%
\end{equation}
The parameter $M$ is the dimension of the quantum system $\bar{A}_{0}$\ that
is communicated from sender to receiver:%
\begin{equation}
M\equiv\left\vert \bar{A}_{0}\right\vert .
\end{equation}

\begin{definition}
[Achievable rate]A rate $R$ is \textit{achievable} for state redistribution of
$\psi_{ABER}$ if for all $\varepsilon\in\left(  0,1\right)  $, $\delta>0$, and
sufficiently large $n$, there exists an $(n,2^{n\left[  R+\delta\right]
},\varepsilon)$ state redistribution protocol.
\end{definition}

\begin{definition}
[Quantum comm.~cost]The \textit{quantum communication cost} $\mathcal{Q}%
(\psi_{ABER})$\ of state redistribution of $\psi_{ABER}$ is equal to the
infimum of all rates which are achievable for redistribution of $\psi_{ABER}$.
\end{definition}

The following theorem from \cite{DY08,YD09}\ gives a precise characterization
of the quantum communication cost:

\begin{theorem}
[\cite{DY08,YD09}]\label{thm:redist-CMI}The quantum communication cost of
state redistribution is equal to half the conditional quantum mutual
information:%
\begin{equation}
\mathcal{Q}(\psi_{ABER})=\frac{1}{2}I(A;B|R)_{\psi}.
\end{equation}

\end{theorem}

The achievability part of the above theorem was simplified in
\cite{PhysRevA.78.030302}, which is the formulation of state redistribution
that we will use to characterize deconstruction cost.

\begin{remark}
\label{rem:unitary-encoding}The results of
\cite{PhysRevA.78.030302,BCT16,DHO14}\ establish that the encoding channel and
decoding channel for state redistribution can be chosen as unitaries, a key
fact that we will use in what follows. Let $U_{A^{n}E^{n}A^{\prime
}\rightarrow\bar{A}_{0}A_{0}\hat{E}^{n}}^{\mathcal{E}}$\ denote the unitary
encoder and $U_{\bar{A}_{0}R^{\prime}R^{n}\rightarrow\hat{A}^{n}\hat{R}%
^{n}R_{0}}^{\mathcal{D}}$ the unitary decoder for these protocols, and note
that the state $\xi_{\hat{A}^{n}B^{n}\hat{E}^{n}\hat{R}^{n}A_{0}R_{0}}$ in
\eqref{eq:final-state-redis-state} can be taken as a pure state as a
consequence. See Figure~\ref{fig:qsr}\ for a depiction of such a state
redistribution protocol.
\end{remark}

\begin{figure}[ptb]
\begin{center}
\includegraphics[
width=3.3399in
]{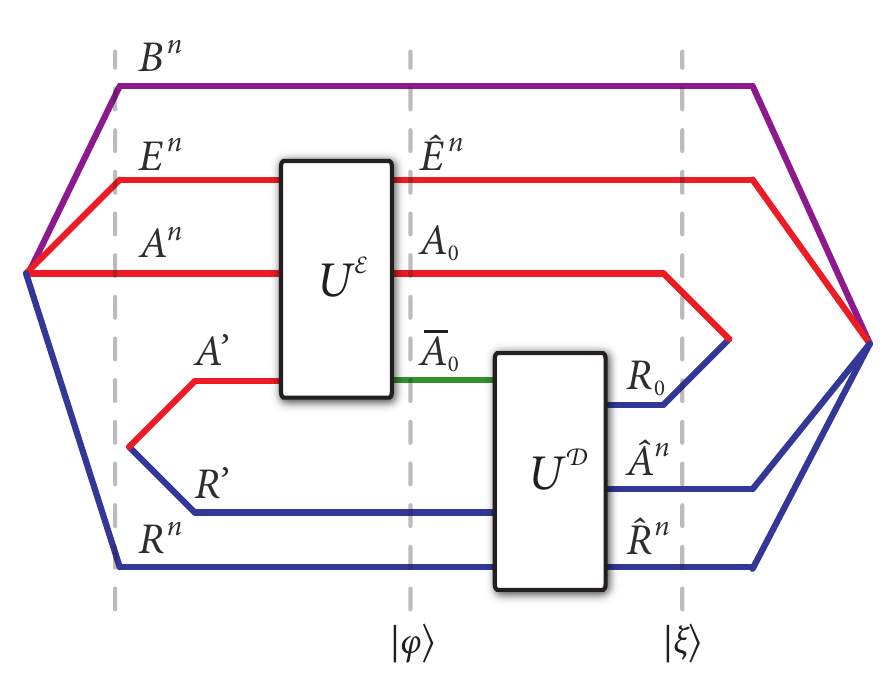}
\end{center}
\caption{Quantum state redistribution with a unitary encoding and decoding. By
using shared entanglement in systems $A^{\prime}$ and $R^{\prime}$ and
noiseless quantum communication of the system $\bar{A}_{0}$, a sender can
transfer her quantum systems $A^{n}$ to a receiver, such that the resulting
state of systems $\hat{A}^{n}B^{n}\hat{R}^{n}\hat{E}^{n}$ has arbitrarily high
fidelity with the initial state of systems $A^{n}B^{n}R^{n}E^{n}$. At the same
time, the protocol generates entanglement in the registers $A_{0}$ and $R_{0}%
$.}%
\label{fig:qsr}%
\end{figure}

We can also quantify the entanglement cost of a quantum state redistribution
protocol. In such a case, for $L\in\mathbb{N}$, we define an
$(n,M,L,\varepsilon)$ quantum state redistribution protocol specified exactly
as given above, except we set%
\begin{equation}
L\equiv|A^{\prime}|/|A_{0}|.
\end{equation}
With this convention, there is an entanglement cost if $L\geq1$ and there is
an entanglement gain if $L\leq1$. A rate pair $(R,E)$ is \textit{achievable}
for state redistribution of $\psi_{ABER}$ if for all $\varepsilon\in\left(
0,1\right)  $, $\delta>0$, and sufficiently large $n$, there exists an
$(n,2^{n\left[  R+\delta\right]  },2^{n\left[  E+\delta\right]  }%
,\varepsilon)$ state redistribution protocol. The achievable rate region of
state redistribution of $\psi_{ABER}$ is equal to the union of all rate pairs
which are achievable for redistribution of $\psi_{ABER}$.

Refs.~\cite{DY08,YD09}\ proved that the rate pair%
\begin{equation}
(I(A;B|R)_{\psi}/2,\left[  I(A;E)_{\psi}-I(A;R)_{\psi}\right]
/2)\label{eq:ach-rate-pair-redist}%
\end{equation}
is achievable and that the optimal rate region is equal to%
\begin{align}
R  & \geq\frac{1}{2}I(A;B|R)_{\psi},\label{eq:redist-region-1}\\
R+E  & \geq H(A|R)_{\psi}.\label{eq:redist-region-2}%
\end{align}
Thus, the rate pair in \eqref{eq:ach-rate-pair-redist} corresponds to an
optimal corner point of the region in
\eqref{eq:redist-region-1}--\eqref{eq:redist-region-2}. The protocol from
\cite{PhysRevA.78.030302} consumes entanglement at a rate equal to
$I(A;E)_{\psi}/2$ and generates entanglement at a rate equal to $I(A;R)_{\psi
}/2$.

\section{State deconstruction protocol\label{sec:deconst-general}}

Here we provide an operational definition for the \textit{deconstruction cost}
of a tripartite state $\rho_{ABE}$. We frame the problem in the formalism of
quantum Shannon theory \cite{W16}, which, as we will show, ultimately leads to
the CQMI being equal to the deconstruction cost after taking a limit. In what
follows, we consider two seemingly different models, called the local unitary
randomizing model and the Landauer--Bennett erasure model. In
Section~\ref{sec:model-equiv}, we show that these two models are in fact
equivalent to each other, in the sense that a protocol from one model can
simulate a protocol from the other, with the same resource consumption and performance.

\subsection{Local unitary randomizing model\label{sec:dec-protocol}}

We begin by defining a state deconstruction protocol in the local unitary
randomizing model. Let $n\in\mathbb{N}$, $M\in\mathbb{N}$, and $\varepsilon
\in\left[  0,1\right]  $. An $(n,M,\varepsilon)$ state deconstruction protocol
consists of an ensemble of $M$ unitaries $\{p_{i},U_{A^{n}A^{\prime}E^{n}}%
^{i}\}_{i=1}^{M}$ that lead to the following local unitary randomizing
channel:%
\begin{equation}
\mathcal{N}_{A^{n}A^{\prime}E^{n}}(\tau_{A^{n}A^{\prime}E^{n}})\equiv\sum
_{i}p_{i}U_{A^{n}A^{\prime}E^{n}}^{i}\tau_{A^{n}A^{\prime}E^{n}}%
(U_{A^{n}A^{\prime}E^{n}}^{i})^{\dag}, \label{eq:LUR-deconstruction}%
\end{equation}
for a density operator $\tau_{A^{n}A^{\prime}E^{n}}$, with system $A^{\prime}$
an auxiliary system. We also refer to such an action as an $\varepsilon
$-\textit{deconstruction operation} and are interested in its action on the
state $\rho_{ABE}^{\otimes n}\otimes\theta_{A^{\prime}}$, where $\theta
_{A^{\prime}}$ is an auxiliary density operator that plays the role of a
catalyst in the sense of\ \cite{MBDRC16}\ to help in the deconstruction task.
The state resulting from a deconstruction operation acting on $\rho
_{ABE}^{\otimes n}\otimes\theta_{A^{\prime}}$\ is as follows:%
\begin{equation}
\omega_{A^{n}A^{\prime}B^{n}E^{n}}\equiv\mathcal{N}_{A^{n}A^{\prime}E^{n}%
}(\rho_{ABE}^{\otimes n}\otimes\theta_{A^{\prime}}). \label{eq:LURed-state}%
\end{equation}

\begin{figure}[ptb]
\begin{center}
\includegraphics[
width=3.0078in
]{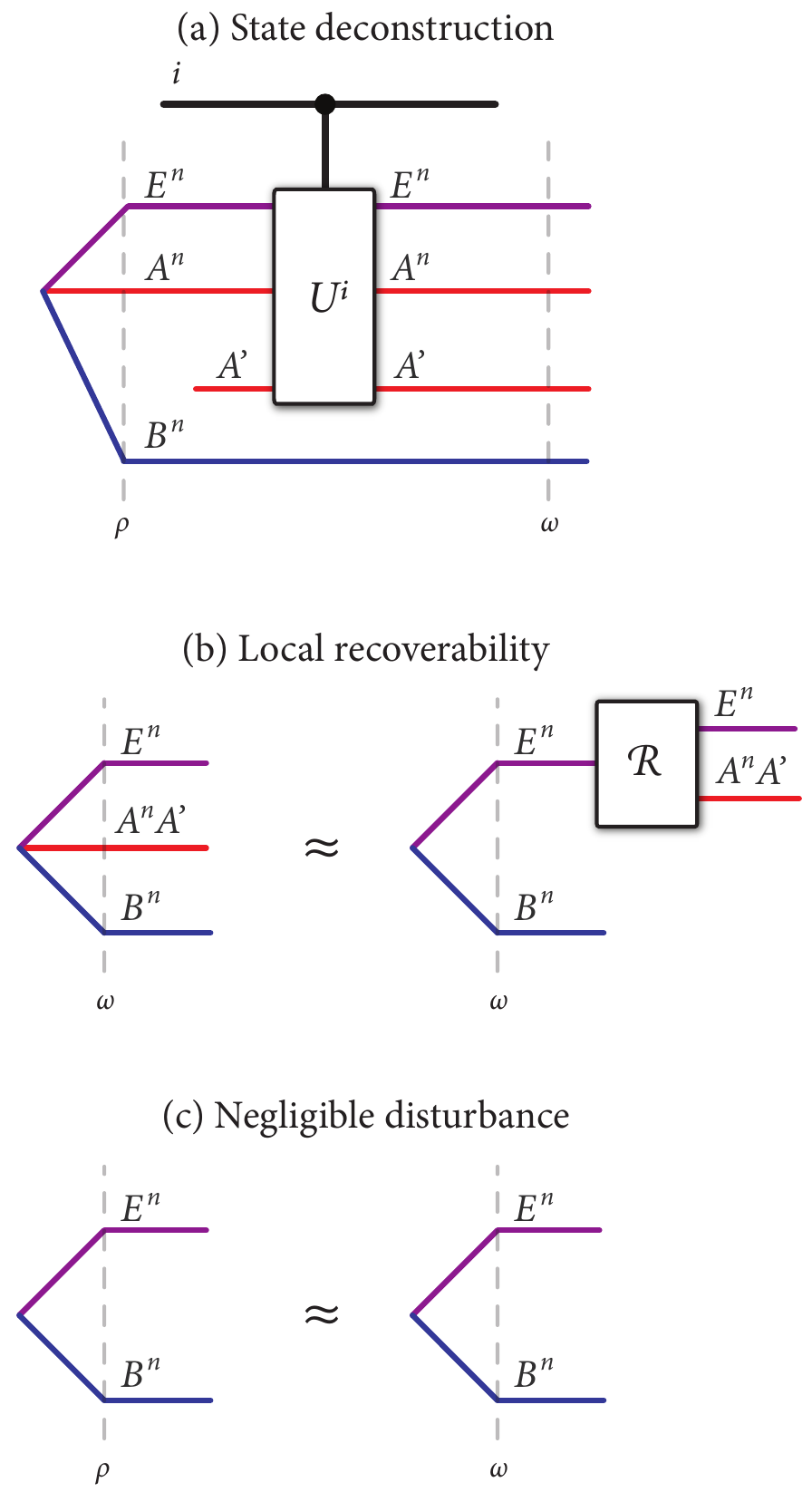}
\end{center}
\caption{Depiction of (a)\ a state deconstruction protocol in the local
unitary randomizing model along with the conditions of (b) local
recoverability and (c) negligible disturbance.}%
\label{fig:deconstruction}%
\end{figure}

We demand for such a deconstruction operation to satisfy the property of
\textit{negligible disturbance} and for the state resulting from the operation
to be \textit{locally recoverable}. In particular, the negligible disturbance
condition means that the deconstruction operation $\mathcal{N}_{A^{n}%
A^{\prime}E^{n}}$ causes little disturbance to the residual state of the
$B^{n}E^{n}$ systems, in the sense that%
\begin{equation}
F(\omega_{B^{n}E^{n}},\rho_{BE}^{\otimes n})\geq1-\varepsilon.
\label{eq:small-disturb-first}%
\end{equation}
The condition of local recoverability means that the resulting state
$\omega_{A^{n}A^{\prime}B^{n}E^{n}}$ is such that the $A^{n}A^{\prime}$
systems are locally recoverable by acting on the $E^{n}$ systems alone. That
is, there exists a recovery channel $\mathcal{R}_{E^{n}\rightarrow
A^{n}A^{\prime}E^{n}}$ such that%
\begin{equation}
F(\omega_{A^{n}A^{\prime}B^{n}E^{n}},\mathcal{R}_{E^{n}\rightarrow
A^{n}A^{\prime}E^{n}}(\omega_{B^{n}E^{n}}))\geq1-\varepsilon.
\label{eq:local-rec}%
\end{equation}
Equivalently, we demand for the following fidelity of recovery to be large:%
\begin{equation}
F(A^{n}A^{\prime};B^{n}|E^{n})_{\omega}\geq1-\varepsilon.
\end{equation}
Figure~\ref{fig:deconstruction}\ depicts a state deconstruction protocol in
the local unitary randomizing model.

\begin{definition}
[Achievable rate]\label{def:achievable-rate}A rate $R$ is \textit{achievable}
for state deconstruction of $\rho_{ABE}$ if for all $\varepsilon\in\left(
0,1\right)  $, $\delta>0$, and sufficiently large $n$, there exists an
$(n,2^{n\left[  R+\delta\right]  },\varepsilon)$ state deconstruction protocol.
\end{definition}

\begin{definition}
[Deconstruction cost]\label{def:dec-cost}The \textit{deconstruction cost}
$\mathcal{D}(A;B|E)_{\rho}$\ of a state $\rho_{ABE}$ is equal to the infimum
of all rates which are achievable for state deconstruction of $\rho_{ABE}$.
\end{definition}

\begin{remark}
\cite[Proposition~35]{BSW14} (refined in \cite[Theorem~11.10.5]{W16}) implies
that the deconstruction cost of $\rho_{ABE}$ is equal to the minimum rate of
noise needed to deconstruct the correlations in $\rho_{ABE}^{\otimes n}$ in
such a way that the resulting state has vanishing normalized CQMI.
Specifically, the state $\omega_{A^{n}A^{\prime}B^{n}E^{n}}$ resulting from an
$\left(  n,M,\varepsilon\right)  $ state deconstruction protocol is such that%
\begin{multline}
\frac{1}{n}I(A^{n}A^{\prime};B^{n}|\hat{E}^{n})_{\omega}\leq2\sqrt
{\varepsilon}\log\left\vert B\right\vert \\
+\frac{1}{n}\left(  1+\sqrt{\varepsilon}\right)  h_{2}(\sqrt{\varepsilon
}/\left[  1+\sqrt{\varepsilon}\right]  ).
\end{multline}

\end{remark}

\begin{remark}
Operational tasks related to state deconstruction were previously explored in
\cite{WSM15}, where a class of \textquotedblleft Markovianizing
operations\textquotedblright\ were defined and subsequently broadened in
\cite{WSM15a,BBW15}. Deconstruction operations are different in that we allow
for a catalyst, a unitary interaction between the $A^{n}E^{n}$ systems and the
catalyst, and we demand for the condition of negligible disturbance to hold.
Whereas our converse (Theorem~\ref{thm:CQMI-lower-bnd}) holds for the model of
\cite{WSM15} as well, the CQMI cannot be achieved: the fact that \cite{WSM15}
does not allow for an interaction with the $E$ systems leads to a strictly
larger optimal rate function based on the Koashi-Imoto
decomposition~\cite{koashi02} (at least for pure states). This proves that the
CQMI cannot be achieved without having access to the $E$ systems. The result
of \cite{WSM15}\ is motivated from questions in distributed computation
\cite{WSM15_2} but has the disadvantage that the Koashi-Imoto decomposition is
not continuous in the state.
\end{remark}

\begin{remark}
In Appendix~\ref{sec:classical-example}, we give a strictly classical example
that demonstrates how the conditional mutual information cannot be achieved
without having access to the $E$ systems. \end{remark}

\subsection{Landauer--Bennett erasure model}

\label{sec:landauer-erasure-model}We can think of deconstruction operations in
an alternative way, akin to the Landauer--Bennett model of erasure
\cite{L61,B73}\ and discussed in \cite[Remark~II.4]{GPW05}, in which we
interact the systems of interest unitarily (reversibly)\ with a catalyst and
subsequently perform a partial trace over some subsystem. The deconstruction
cost in this case is then related to the size of the system that we trace out.
In this alternative model, we define a deconstruction operation $\mathcal{N}%
_{A^{n}E^{n}\rightarrow A_{1}^{\prime}\hat{E}^{n}}$ as%
\begin{align}
\omega_{A_{1}^{\prime}B^{n}\hat{E}^{n}} &  \equiv\mathcal{N}_{A^{n}%
E^{n}\rightarrow A_{1}^{\prime}\hat{E}^{n}}(\rho_{ABE}^{\otimes n}%
)\label{eq:omega-state-LB}\\
&  \equiv\operatorname{Tr}_{A_{2}^{\prime}}\{\mathcal{U}_{A^{n}E^{n}A^{\prime
}\rightarrow A_{1}^{\prime}A_{2}^{\prime}\hat{E}^{n}}(\rho_{ABE}^{\otimes
n}\otimes\theta_{A^{\prime}})\},\label{eq:decons-op-landauer}%
\end{align}
with $\theta_{A^{\prime}}$ an arbitrary ancilla state and $\mathcal{U}%
_{A^{n}E^{n}A^{\prime}\rightarrow A_{1}^{\prime}A_{2}^{\prime}\hat{E}^{n}}$ a
unitary quantum channel. An $(n,M,\varepsilon)$ deconstruction protocol in
this case has $n$ defined again as the number of copies of $\rho_{ABE}$ and
$\varepsilon$ defined via \eqref{eq:small-disturb-first} and $F(A_{1}^{\prime
};B^{n}|\hat{E}^{n})_{\omega}\geq1-\varepsilon$.\ However, in this
Landauer--Bennett erasure model, we take $M$ defined as%
\begin{equation}
M\equiv\left\vert A_{2}^{\prime}\right\vert ^{2}.\label{eq:rate-Markov}%
\end{equation}
In this model, we take the convention of squaring the dimension of the removed
system $\left\vert A_{2}^{\prime}\right\vert ^{2}$ when calculating $M$,
because we are interested in measuring the amount of noise needed to remove
the $A_{2}^{\prime}$ system (i.e., the amount of noise needed to physically
implement a partial trace). One way to do so is to apply a randomizing channel
of the following form, which realizes a partial trace:%
\begin{multline}
\frac{1}{\left\vert A_{2}^{\prime}\right\vert ^{2}}\sum_{i=1}^{\left\vert
A_{2}^{\prime}\right\vert ^{2}}V_{A_{2}^{\prime}}^{i}\mathcal{U}_{A^{n}%
E^{n}A^{\prime}\rightarrow A_{1}^{\prime}A_{2}^{\prime}\hat{E}^{n}}(\rho
_{ABE}^{\otimes n}\otimes\theta_{A^{\prime}})(V_{A_{2}^{\prime}}^{i})^{\dag
}\label{eq:HW-partial-trace}\\
=\pi_{A_{2}^{\prime}}\otimes\operatorname{Tr}_{A_{2}^{\prime}}\{\mathcal{U}%
_{A^{n}E^{n}A^{\prime}\rightarrow A_{1}^{\prime}A_{2}^{\prime}\hat{E}^{n}%
}(\rho_{ABE}^{\otimes n}\otimes\theta_{A^{\prime}})\},
\end{multline}
where $\{V_{A_{2}^{\prime}}^{i}\}_{i=1}^{\left\vert A_{2}^{\prime}\right\vert
^{2}}$ is a unitary one-design and $\pi_{A_{2}^{\prime}}\equiv I_{A_{2}^{\prime}}/|A_{2}^{\prime}|$ is
the maximally mixed state. It is known that $\left\vert A_{2}^{\prime
}\right\vert ^{2}$ unitaries are necessary and sufficient for physically
implementing a partial trace in the above sense \cite{AMTW00}.

We can then define achievable rates and the deconstruction cost for this
alternative model just as in Definitions~\ref{def:achievable-rate}\ and
\ref{def:dec-cost}. This model might seem as if it is slightly different from
the local unitary randomizing one, but we show in the next section that they
are equivalent and thus lead to the same deconstruction cost.

\subsection{Equivalence of the two models\label{sec:model-equiv}}

In this section, we show that the local unitary randomizing model and the
Landauer--Bennett erasure models are equivalent, in the sense that they can
simulate one another with the same performance and resource consumption. This
equivalence was shown for a special case in \cite{MBDRC16}, and here we
generalize the argument to the settings considered in this paper. As a
consequence of our simulation argument, there is no need to consider two
different notions of deconstruction cost, since the simulation argument
implies that the costs are in fact the same.

First, we show that the local unitary randomizing model can simulate the
Landauer--Bennett erasure model. To this end, suppose that we are given a
catalyst state $\theta_{A^{\prime}}$ and an interaction unitary $U_{A^{n}%
E^{n}A^{\prime}\rightarrow A_{1}^{\prime}A_{2}^{\prime}\hat{E}^{n}}$, such
that the Landauer--Bennett erasure deconstruction operation is as given in
\eqref{eq:decons-op-landauer}. We can simulate such an operation by choosing
an ensemble of unitaries to be as follows:
\begin{equation}
\{1/|A_{2}^{\prime}|^{2},W_{A^{n}E^{n}A^{\prime}\rightarrow A_{1}^{\prime
}A_{2}^{\prime}\hat{E}^{n}}^{i}\}_{i=1}^{|A_{2}^{\prime}|^{2}},
\label{eq:simulating-ensemble}%
\end{equation}
where%
\begin{equation}
W_{A^{n}E^{n}A^{\prime}\rightarrow A_{1}^{\prime}A_{2}^{\prime}\hat{E}^{n}%
}^{i}\equiv V_{A_{2}^{\prime}}^{i}U_{A^{n}E^{n}A^{\prime}\rightarrow
A_{1}^{\prime}A_{2}^{\prime}\hat{E}^{n}}%
\end{equation}
and $\{V_{A_{2}^{\prime}}^{i}\}_{i=1}^{\left\vert A_{2}^{\prime}\right\vert
^{2}}$ is a set of Heisenberg--Weyl unitaries that realize a partial
trace.\ The result is that a local unitary randomizing channel in
\eqref{eq:LUR-deconstruction} formed from the ensemble in
\eqref{eq:simulating-ensemble} can realize the deconstruction operation in
\eqref{eq:decons-op-landauer}:%
\begin{multline}
\frac{1}{|A_{2}^{\prime}|^{2}}\sum_{i}W^{i}(\rho_{ABE}^{\otimes n}%
\otimes\theta_{A^{\prime}})W^{i\dag}\\
=\pi_{A_{2}^{\prime}}\otimes\operatorname{Tr}_{A_{2}^{\prime}}\{\mathcal{U}%
_{A^{n}E^{n}A^{\prime}\rightarrow A_{1}^{\prime}A_{2}^{\prime}\hat{E}^{n}%
}(\rho_{ABE}^{\otimes n}\otimes\theta_{A^{\prime}})\}\\
\equiv\pi_{A_{2}^{\prime}}\otimes\omega_{A_{1}^{\prime}B^{n}\hat{E}^{n}}.
\end{multline}
Both the negligible disturbance and the local recoverability conditions hold
with the same quality as in the original protocol. This is clear for the
negligible disturbance condition, and to see it for the local recoverability
condition, we can invoke a special case of the multiplicativity of fidelity of
recovery with respect to tensor-product states \cite{BT15}:%
\begin{equation}
F(A_{1}^{\prime}A_{2}^{\prime};B^{n}|\hat{E}^{n})_{\pi\otimes\omega}%
=F(A_{1}^{\prime};B^{n}|\hat{E}^{n})_{\omega}.
\end{equation}

Showing the other simulation requires a bit more effort. To this end, consider
an arbitrary ensemble of unitaries $\{p_{i},U_{A^{n}A^{\prime}E^{n}}%
^{i}\}_{i=1}^{M}$ and an ancilla $\theta_{A^{\prime}}$. We need to show how it
is possible to simulate the effect of a local unitary randomizing channel of
the form in \eqref{eq:LUR-deconstruction}\ built from this ensemble, by
bringing in an ancilla state, performing a global unitary, and ending with a
partial trace. We take the ancilla to be the following state:%
\begin{equation}
\pi_{S_{A}}\otimes\pi_{T_{A}}\otimes\sum_{i=1}^{M}p_{i}|i\rangle\langle
i|_{\hat{M}_{A}}\otimes\theta_{A^{\prime}},
\end{equation}
where $S_{A}$ and $T_{A}$ are quantum systems each having dimension equal to
$\sqrt{M}$. (Note that if $\sqrt{M}$ is not an integer, then we can
\textquotedblleft zero-pad\textquotedblright\ the probability distribution
$\{p_{i}\}$ such that its cardinality becomes a power of two---this has the
negligible effect of incrementing by one the number of bits needed to describe
the indices $i$ corresponding to the entries of the probability distribution
$\{p_{i}\}$ and at the same time ensures that $\sqrt{M}$ is an integer). It is
helpful to recall the following equality:%
\begin{equation}
\pi_{S_{A}}\otimes\pi_{T_{A}}=\frac{1}{M}\sum_{j,k}|\Phi^{j,k}\rangle
\langle\Phi^{j,k}|_{S_{A}T_{A}},
\end{equation}
where $\{|\Phi^{j,k}\rangle_{S_{A}T_{A}}\}$ denotes the Bell basis reviewed in
Section~\ref{sec:basics}. We take the unitary interaction between the ancilla
systems $S_{A}T_{A}\hat{M}_{A}A^{\prime}$\ and the data systems $A^{n}E^{n}$
to be a serial concatenation of the following two controlled unitaries:%
\begin{align}
&  \sum_{i}|i\rangle\langle i|_{\hat{M}_{A}}\otimes U_{A^{n}A^{\prime}E^{n}%
}^{i},\label{eq:1st-contr-U}\\
&  \sum_{j,k}|\Phi^{j,k}\rangle\langle\Phi^{j,k}|_{S_{A}T_{A}}\otimes
X_{\hat{M}_{A}}^{(j-1)\cdot d+k}. \label{eq:2nd-contr-U}%
\end{align}
The state resulting from applying these two controlled unitaries sequentially
(\eqref{eq:1st-contr-U} and then \eqref{eq:2nd-contr-U}) to the systems
$S_{A}T_{A}\hat{M}_{A}A^{n}E^{n}F$ is as follows:%
\begin{multline}
\frac{1}{M}\sum_{j,k}|\Phi^{j,k}\rangle\langle\Phi^{j,k}|_{S_{A}T_{A}}\\
\otimes\sum_{i}p_{i}X_{\hat{M}_{A}}^{(j-1)\cdot d+k}|i\rangle\langle
i|_{\hat{M}_{A}}[X_{\hat{M}_{A}}^{(j-1)\cdot d+k}]^{\dag}\\
\otimes U_{A^{n}A^{\prime}E^{n}}^{i}(\rho_{ABE}^{\otimes n}\otimes
\theta_{A^{\prime}})(U_{A^{n}A^{\prime}E^{n}}^{i})^{\dag}.
\end{multline}
After tracing over the $S_{A}$ register, which requires $\log M$ bits of noise
according to our convention in \eqref{eq:rate-Markov}, the state becomes as
follows:%
\begin{multline}
\pi_{T_{A}}\otimes\pi_{\hat{M}_{A}}\otimes\sum_{i}p_{i}U_{A^{n}A^{\prime}%
E^{n}}^{i}(\rho_{ABE}^{\otimes n}\otimes\theta_{A^{\prime}})(U_{A^{n}%
A^{\prime}E^{n}}^{i})^{\dag}\\
\equiv\pi_{T_{A}}\otimes\pi_{\hat{M}_{A}}\otimes\omega_{A^{n}A^{\prime}%
B^{n}E^{n}}.
\end{multline}
One can verify this explicitly, or see that it follows intuitively from a
cascade: tracing over system $S_{A}$ has the effect of \textquotedblleft
forgetting\textquotedblright\ $j$ and $k$, which has the effect of randomizing
the classical system $\hat{M}_{A}$ with a uniform mixture of the shift
operators $X_{\hat{M}_{A}}^{(j-1)\cdot d+k}$, which in turn has the effect of
\textquotedblleft forgetting\textquotedblright\ $i$, which then applies the
local unitary randomizing channel to the systems $A^{n}A^{\prime}E^{n}$. Both
the negligible disturbance and the local recoverability conditions hold with
the same quality as in the original protocol. This is clear for the negligible
disturbance condition, and to see it for the local recoverability condition,
we can invoke a special case of the multiplicativity of fidelity of recovery
with respect to tensor-product states \cite{BT15}:%
\begin{equation}
F(T_{A}\hat{M}_{A}A^{n}A^{\prime};B^{n}|E^{n})_{\pi\otimes\pi\otimes\omega
}=F(A^{n}A^{\prime};B^{n}|E^{n})_{\omega}.
\end{equation}

\section{Deconstruction cost is lower bounded by
CQMI\label{sec:CQMI-lower-bound}}

In this section, we prove that the deconstruction cost of a tripartite state
$\rho_{ABE}$\ is lower bounded by its conditional quantum mutual information
$I(A;B|E)_{\rho}$. We prove such a converse theorem in the Landauer--Bennett
erasure model. By the simulation argument given in
Section~\ref{sec:model-equiv}, this theorem also serves as a converse bound
for deconstruction cost in the local unitary randomizing model. For the
interested reader, Appendix~\ref{sec:alt-conv-proof}\ offers two alternative
converse proofs for optimality of the deconstruction cost in the local unitary
randomizing model. One of them has a flavor similar to the converse proof
given below, and the other is similar to those from prior works
\cite{GPW05,BBW15,WSM15a}.

\begin{theorem}
\label{thm:CQMI-lower-bnd}The conditional quantum mutual information
$I(A;B|E)_{\rho}$ of a tripartite state $\rho_{ABE}$\ is a lower bound on its
deconstruction cost $\mathcal{D}(A;B|E)_{\rho}$:%
\begin{equation}
I(A;B|E)_{\rho}\leq\mathcal{D}(A;B|E)_{\rho}.
\end{equation}

\end{theorem}

\begin{proof}
To prove this theorem, we employ entropy inequalities and properties of CQMI.
Consider a general $(n,M,\varepsilon)$ Landauer--Bennett state deconstruction
protocol as outlined in Section~\ref{sec:landauer-erasure-model}. Then the
following chain of inequalities holds
\begin{align}
&  nI(A;B|E)_{\rho}\nonumber\\
&  =I(A^{n};B^{n}|E^{n})_{\rho^{\otimes n}}\nonumber\\
&  =H(B^{n}|E^{n})_{\rho^{\otimes n}}-H(B^{n}|A^{n}E^{n})_{\rho^{\otimes n}%
}\nonumber\\
&  =H(B^{n}|E^{n})_{\rho^{\otimes n}}-H(B^{n}|A^{n}A^{\prime}E^{n}%
)_{\rho^{\otimes n}\otimes\theta}\nonumber\\
&  \leq H(B^{n}|\hat{E}^{n})_{\omega}+f(n,\varepsilon)-H(B^{n}|A_{1}^{\prime
}A_{2}^{\prime}\hat{E}^{n})_{\mathcal{U}(\rho^{\otimes n}\otimes\theta
)}\nonumber\\
&  \leq H(B^{n}|\hat{E}^{n})_{\omega}+f(n,\varepsilon)-H(B^{n}|A_{1}^{\prime
}\hat{E}^{n})_{\omega}+2\log_{2}\left\vert A_{2}^{\prime}\right\vert
\nonumber\\
&  =2\log_{2}\left\vert A_{2}^{\prime}\right\vert +I(A_{1}^{\prime};B^{n}%
|\hat{E}^{n})_{\omega}+f(n,\varepsilon)\nonumber\\
&  \leq2\log_{2}\left\vert A_{2}^{\prime}\right\vert +g(n,\varepsilon
)+f(n,\varepsilon). \label{eq:CQMI-converse-1st-block}%
\end{align}
The first equality follows because the CQMI is additive with respect to
tensor-product states. The second equality follows from the definition of
CQMI. The third equality follows because the conditional entropy is invariant
with respect to tensoring in a product state to be part of the conditioning
system. The first inequality follows because the conditional entropy is
invariant with respect to a local unitary acting on the conditioning system:%
\begin{equation}
H(B^{n}|A^{n}A^{\prime}E^{n})_{\rho^{\otimes n}\otimes\theta}=H(B^{n}%
|A_{1}^{\prime}A_{2}^{\prime}\hat{E}^{n})_{\mathcal{U}(\rho^{\otimes n}%
\otimes\theta)}.
\end{equation}
Also, we have applied the negligible disturbance condition from
\eqref{eq:small-disturb-first}, the Fuchs-van-de-Graaf inequalities in
\eqref{eq:F-v-d-G-ineq}, and the continuity of conditional entropy
\cite{AF04,Winter15}, with%
\begin{equation}
f(n,\varepsilon)=2\sqrt{\varepsilon}n\log\left\vert B\right\vert
+(1+\sqrt{\varepsilon})h_{2}(\sqrt{\varepsilon}/[1+\sqrt{\varepsilon}]).
\label{eq:f-func-continuity-AFW}%
\end{equation}
The second inequality follows from a rewriting and applying a dimension bound
for CQMI\ (see, e.g., \cite[Exercise~11.7.9]{W16}):%
\begin{multline}
H(B^{n}|A_{1}^{\prime}\hat{E}^{n})_{\mathcal{U}(\rho^{\otimes n}\otimes
\theta)}-H(B^{n}|A_{1}^{\prime}A_{2}^{\prime}\hat{E}^{n})_{\mathcal{U}%
(\rho^{\otimes n}\otimes\theta)}\\
=I(B^{n};A_{2}^{\prime}|A_{1}^{\prime}\hat{E}^{n})_{\mathcal{U}(\rho^{\otimes
n}\otimes\theta)}\leq2\log_{2}\left\vert A_{2}^{\prime}\right\vert .
\end{multline}
The last equality follows from the definition of CQMI. The final inequality
follows by applying the local recoverability condition $F(A_{1}^{\prime}%
;B^{n}|\hat{E}^{n})_{\omega}\geq1-\varepsilon$ and because locally recoverable
states have small CQMI as reviewed in \eqref{eq:loc-rec-low-CQMI}. In
particular, we can take%
\begin{equation}
g(n,\varepsilon)\equiv2n\sqrt{\varepsilon}\log\left\vert B\right\vert +\left(
1+\sqrt{\varepsilon}\right)  h_{2}(\sqrt{\varepsilon}/\left[  1+\sqrt
{\varepsilon}\right]  ). \label{eq:loc-rec-CQMI-cont}%
\end{equation}
Thus, recalling our convention that $M=\left\vert A_{2}^{\prime}\right\vert
^{2}$, we conclude that the following bound holds for any $\left(
n,M,\varepsilon\right)  $ state deconstruction protocol:%
\begin{equation}
I(A;B|E)_{\rho}\leq\frac{1}{n}\log_{2}M+\frac{1}{n}\left[  g(n,\varepsilon
)+f(n,\varepsilon)\right]  .
\end{equation}
By taking the limit as $n\rightarrow\infty$, then $\varepsilon\rightarrow0$,
and applying definitions, we can conclude the inequality $I(A;B|E)_{\rho}%
\leq\mathcal{D}(A;B|E)_{\rho}$.
\end{proof}

\section{From state redistribution to state deconstruction}

\label{sec:simulation-via-redis}To show that the deconstruction cost is
achievable (i.e., that $\mathcal{D}(A;B|E)_{\rho}\leq I(A;B|E)_{\rho}$), we
employ the quantum state redistribution protocol, reviewed in
Section~\ref{sec:QSR-review}. We begin by proving that a state redistribution
protocol implies the existence of a state deconstruction protocol.

\begin{theorem}
\label{thm:redis-sim-decon}An $\left(  n,M,\varepsilon\right)  $ protocol for
state redistribution of a four-system pure state $\psi_{ABER}$, as specified
in Section~\ref{sec:QSR-review}, realizes an $\left(  n,M^{2},4\varepsilon
\right)  $ protocol for state deconstruction of $\rho_{ABE}=\operatorname{Tr}%
_{R}\{\psi_{ABER}\}$, as specified in Section~\ref{sec:deconst-general}.
\end{theorem}

\begin{proof}
Let $\psi_{ABER}$ be a purification of $\rho_{ABE}$. Given is an $\left(
n,M,\varepsilon\right)  $ state redistribution protocol, which by
Remark~\ref{rem:unitary-encoding}\ means that there is a unitary encoder
$U_{A^{n}E^{n}A^{\prime}\rightarrow\bar{A}_{0}A_{0}\hat{E}^{n}}^{\mathcal{E}}%
$\ and a unitary decoder $U_{\bar{A}_{0}R^{\prime}R^{n}\rightarrow\hat{A}%
^{n}\hat{R}^{n}R_{0}}^{\mathcal{D}}$ satisfying \eqref{eq:good-QSR}. We will
show the existence of an $\left(  n,M^{2},4\varepsilon\right)  $ protocol for
state deconstruction of $\rho_{ABE}$ in the Landauer--Bennett erasure model.
By the monotonicity of fidelity with respect to partial trace over the systems
$\hat{A}^{n}\hat{R}^{n}R_{0}$ \cite[Lemma~9.2.1]{W16}, Eq.~\eqref{eq:good-QSR}
implies that%
\begin{equation}
F(\xi_{A_{0}B^{n}\hat{E}^{n}},\pi_{A_{0}}\otimes\rho_{B\hat{E}}^{\otimes
n})\geq1-\varepsilon. \label{eq:small-disturb}%
\end{equation}

In our protocol for state deconstruction, we take the deconstruction operation
to be

\begin{enumerate}
\item tensoring in the maximally mixed state $\pi_{A^{\prime}}$,

\item application of the unitary $U_{A^{n}E^{n}A^{\prime}\rightarrow\bar
{A}_{0}A_{0}\hat{E}^{n}}^{\mathcal{E}}$,

\item a partial trace over the $\bar{A}_{0}$ system.
\end{enumerate}

\noindent Let%
\begin{align}
\omega_{A_{0}B^{n}\hat{E}^{n}}  &  \equiv\operatorname{Tr}_{\bar{A}_{0}%
}\{U^{\mathcal{E}}(\rho_{AB\hat{E}}^{\otimes n}\otimes\pi_{A^{\prime}%
})U^{\mathcal{E}\dag}\}\\
&  =\xi_{A_{0}B^{n}\hat{E}^{n}},
\end{align}
where $U^{\mathcal{E}}\equiv U_{A^{n}E^{n}A^{\prime}\rightarrow\bar{A}%
_{0}A_{0}\hat{E}^{n}}^{\mathcal{E}}$.

Now we show that the protocol satisfies the requirements of negligible
disturbance and local recoverability, as outlined in
Section~\ref{sec:landauer-erasure-model}. The condition of negligible
disturbance follows directly from \eqref{eq:small-disturb}, after a partial
trace over system $A_{0}$, because%
\begin{equation}
\xi_{B^{n}\hat{E}^{n}}=\operatorname{Tr}_{A_{0}}\{\omega_{A_{0}B^{n}\hat
{E}^{n}}\}.
\end{equation}

The condition of local recoverability follows rather directly as well from
\eqref{eq:small-disturb}. If the system $A_{0}$ is lost, then the remaining
state is $\xi_{B^{n}\hat{E}^{n}}$. We can then take the recovery channel to
merely tensor in a maximally mixed state $\pi_{A_{0}}$, and
\eqref{eq:small-disturb} guarantees that the resulting state is close to the
original one. Indeed, by employing the fact that $\sqrt{1-F(\rho,\sigma)}$ is
a distance measure \cite{GLN05} and thus obeys the triangle inequality, we
find that%
\begin{multline}
\sqrt{1-F(\xi_{A_{0}B^{n}\hat{E}^{n}},\pi_{A_{0}}\otimes\xi_{B^{n}\hat{E}^{n}%
})}\label{eq:triangle-ineq-1-F}\\
\leq\sqrt{1-F(\xi_{A_{0}B^{n}\hat{E}^{n}},\pi_{A_{0}}\otimes\rho_{B\hat{E}%
}^{\otimes n})}\\
+\sqrt{1-F(\pi_{A_{0}}\otimes\rho_{B\hat{E}}^{\otimes n},\pi_{A_{0}}\otimes
\xi_{B^{n}\hat{E}^{n}})}\leq2\sqrt{\varepsilon},
\end{multline}
where the second inequality follows from \eqref{eq:small-disturb} and the fact
that%
\begin{equation}
F(\pi_{A_{0}}\otimes\rho_{B\hat{E}}^{\otimes n},\pi_{A_{0}}\otimes\xi
_{B^{n}\hat{E}^{n}})=F(\rho_{B\hat{E}}^{\otimes n},\xi_{B^{n}\hat{E}^{n}}%
)\geq1-\varepsilon.
\end{equation}
Then we find that%
\begin{equation}
F(\xi_{A_{0}B^{n}\hat{E}^{n}},\pi_{A_{0}}\otimes\xi_{B^{n}\hat{E}^{n}}%
)\geq1-4\varepsilon,\label{eq:cond-decoupled-last}%
\end{equation}
concluding the proof.
\end{proof}

The following is then a direct corollary of Theorem~\ref{thm:redis-sim-decon},
the definitions of state redistribution and state deconstruction in
Sections~\ref{sec:QSR-review}\ and \ref{sec:deconst-general}, respectively,
and Theorem~\ref{thm:redist-CMI}:

\begin{corollary}
\label{cor:CQMI-upper-decons}The deconstruction cost $\mathcal{D}%
(A;B|E)_{\rho}$ of a tripartite state $\rho_{ABE}$\ is bounded from above by
its CQMI $I(A;B|E)_{\rho}$:%
\begin{equation}
\mathcal{D}(A;B|E)_{\rho}\leq I(A;B|E)_{\rho}.
\end{equation}

\end{corollary}

As a consequence of Theorem~\ref{thm:CQMI-lower-bnd}\ and
Corollary~\ref{cor:CQMI-upper-decons}, we can conclude one of our main
results, as stated at the beginning of Section~\ref{sec:main-result}.

\begin{theorem}
\label{thm:main}The deconstruction cost $\mathcal{D}(A;B|E)_{\rho}$ of a
tripartite state $\rho_{ABE}$ is equal to its CQMI $I(A;B|E)_{\rho}$:%
\begin{equation}
\mathcal{D}(A;B|E)_{\rho}=I(A;B|E)_{\rho}.
\end{equation}

\end{theorem}

\subsection{Special case of classical side information}

The state deconstruction protocol can be simplified in the case that the
system $E$ is classical. If this is the case, then the tripartite state
$\rho_{ABE}$\ has the form $\rho_{ABE}=\sum_{e}p_{E}(e)\rho_{AB}^{e}%
\otimes|e\rangle\langle e|_{E}$, where $p_{E}(e)$ is a probability
distribution, $\{\rho_{AB}^{e}\}$ is a set of states, $\{|e\rangle_{E}\}$ is
an orthonormal basis, and the symbol~$e$ is chosen from an
alphabet~$\mathcal{E}$. In this case, we have
\begin{align}
\rho_{ABE}^{\otimes n} &  =\sum_{e^{n}}p_{E^{n}}(e^{n})\rho_{A^{n}B^{n}%
}^{e^{n}}\otimes|e^{n}\rangle\langle e^{n}|_{E^{n}},\\
p_{E^{n}}(e^{n}) &  \equiv\prod\limits_{j=1}^{n}p_{E}(e_{j}),\\
\rho_{A^{n}B^{n}}^{e^{n}} &  =\rho_{A_{1}B_{1}}^{e_{1}}\otimes\cdots
\otimes\rho_{A_{n}B_{n}}^{e_{n}},\\
|e^{n}\rangle_{E^{n}} &  =|e_{1}\rangle_{E_{1}}\otimes\cdots\otimes
|e_{n}\rangle_{E_{n}}.
\end{align}
The protocol proceeds by performing a typical subspace measurement of the
systems $E^{n}$ \cite{W16}, keeping only the classical sequences which are
typical (i.e., those with empirical distribution close to the distribution
$p_{E}$). All such sequences can be partitioned into $\left\vert
\mathcal{E}\right\vert $\ blocks, each consisting of the same symbol
$e\in\mathcal{E}$\ and with length $\approx np_{E}(e)$. For each block, we
then employ the erasure of correlations protocol from \cite{GPW05}, which
implies that $\approx np_{E}(e)I(A;B)_{\rho^{e}}$ bits of noise are used to
erase the correlations in a given block. Thus the total rate of noise needed
in this case is equal to $\sum_{e}p_{E}(e)I(A;B)_{\rho^{e}}=I(A;B|E)_{\rho}$.
The above protocol falls into the class of deconstruction operations because
it causes zero disturbance to the marginal state on systems $B^{n}E^{n}$.
Furthermore, the state afterward is locally recoverable. The result of the
erasure of correlations protocol is to produce a state close to one of the
form $\sum_{e^{n}}p_{E^{n}}(e^{n})\omega_{A^{n}}^{e^{n}}\otimes\omega_{B^{n}%
}^{e^{n}}\otimes|e^{n}\rangle\langle e^{n}|_{E^{n}}$, for which the recovery
procedure is clear:\ if system $A^{n}$ gets lost, look in system $E^{n}$ for
the classical sequence $e^{n}$ and then prepare the state $\omega_{A^{n}%
}^{e^{n}}$ in the $A$ systems.

One further observation is that the protocol given above does not require
access to a catalyst in this special case. It is largely open to determine
whether a catalyst is actually needed in the fully quantum case (i.e., when
the $E$ system does not admit a classical description).


\section{Conditional erasure}

\label{sec:cond-erasure}We now turn to conditional erasure and begin by
providing an operational definition of a conditional erasure protocol, doing
so in the Landauer--Bennett erasure model from
Section~\ref{sec:landauer-erasure-model}. There are some similarities between
state deconstruction and conditional erasure, but in our development for
conditional erasure, we also quantify the rate of noise being consumed or
generated by a given protocol. To this end, we distinguish and quantify two
types of noise, which we call active noise and passive noise.

Active noise is synonymous with a partial trace in the Landauer--Bennett
erasure model from Section~\ref{sec:landauer-erasure-model}. The amount of
active noise being applied in the operation in \eqref{eq:decons-op-landauer}
is equal to $M=\left\vert A_{2}^{\prime}\right\vert ^{2}$ and the rate of
active noise is equal to $\left[  \log_{2}M\right]  /n$. We use the term
active noise to describe this kind of noise because one needs to apply a
physical procedure, consisting of local randomizing unitaries, in order to
implement an active noise operation and realize a partial trace.

Passive noise is synonymous with a catalyst that is brought in to help
accomplish an erasure task. Here, we consider passive noise as a resource and
quantify it as follows: the amount of passive noise is equal to the dimension
$d$\ of the catalyst and the rate of passive noise is equal to $\left[
\log_{2}d\right]  /n$. We use the term passive noise to describe this kind of
noise because one only needs to bring in a maximally mixed state as a
resource: there is no need to apply local randomizing unitaries to create
passive noise. It is also clear that active noise can create passive noise but
not vice versa.

With these notions in mind, we can now define a conditional erasure protocol.
Let $n\in\mathbb{N}$, $M,L\in\mathbb{N}$, and $\varepsilon\in\left[
0,1\right]  $. An $(n,M,L,\varepsilon)$ conditional erasure protocol consists
of a unitary quantum channel $\mathcal{U}_{A^{n}E^{n}A^{\prime}\rightarrow
A_{1}^{\prime}A_{2}^{\prime}\hat{E}^{n}}$ and an auxiliary catalyst state
$\pi_{A^{\prime}}$, which is maximally mixed. The state at the end of the
protocol is $\omega_{A_{1}^{\prime}B^{n}\hat{E}^{n}}$, as given in
\eqref{eq:omega-state-LB}. The parameter $M$ is equal to $\left\vert
A_{2}^{\prime}\right\vert ^{2}$ as before. We require that a conditional
erasure protocol satisfies the property of negligible disturbance, as
specified in \eqref{eq:small-disturb-first}. We also require that the
resulting state $\omega_{A_{1}^{\prime}B^{n}\hat{E}^{n}}$ is such that the
$A_{1}^{\prime}$ system is decoupled from the $BE$ systems, in the sense that%
\begin{equation}
F(\omega_{A_{1}^{\prime}B^{n}\hat{E}^{n}},\pi_{A_{1}^{\prime}}\otimes
\omega_{B^{n}\hat{E}^{n}})\geq1-\varepsilon,\label{eq:cond-decoupled}%
\end{equation}
where $\pi_{A_{1}^{\prime}}$ is a maximally mixed state. We take the parameter%
\begin{equation}
L=\left(  \left\vert A^{\prime}\right\vert /\left\vert A_{1}^{\prime
}\right\vert \right)  ^{2},\label{eq:L-param-cond-erasure}%
\end{equation}
or equivalently, $\log_{2}L=2\left[  \log_{2}\left\vert A^{\prime}\right\vert
-\log_{2}\left\vert A_{1}^{\prime}\right\vert \right]  $. The parameter
$L$\ thus quantifies the gain or consumption of passive noise in a conditional
erasure protocol. If passive noise is gained in a conditional erasure
protocol, then it can be used as a resource for a future erasure task.

We can see by inspecting \eqref{eq:cond-decoupled}\ that conditional erasure
achieves the task of state deconstruction, with the local recovery channel
taken to be a preparation of the state $\pi_{A_{1}^{\prime}}$ after the system
$A_{1}^{\prime}$ of $\omega_{A_{1}^{\prime}B^{n}\hat{E}^{n}}$ is lost.

\subsection{Conditional erasure is equivalent to state redistribution}

In this section, we show that the task of conditional erasure is equivalent to
state redistribution, in the sense that the existence of a conditional erasure
protocol implies the existence of a state redistribution protocol and vice
versa. We begin with the following implication:

\begin{theorem}
\label{thm:redis-sim-cond-erasure}An $\left(  n,M,L,\varepsilon\right)  $
protocol for state redistribution of a four-system pure state $\psi_{ABER}$,
as specified in Section~\ref{sec:QSR-review}, realizes an $\left(
n,M^{2},L^{2},4\varepsilon\right)  $ conditional erasure protocol of
$\rho_{ABE}=\operatorname{Tr}_{R}\{\psi_{ABER}\}$, as specified in
Section~\ref{sec:cond-erasure}.
\end{theorem}

\begin{proof}
A proof of this theorem directly follows along the lines given in the proof of
Theorem~\ref{thm:redis-sim-decon}. Following the proof there, we arrive at
\eqref{eq:cond-decoupled-last}, which is equivalent to the desired condition
in \eqref{eq:cond-decoupled}. The parameter $L$ for the state redistribution
protocol is equal to $\left\vert A^{\prime}\right\vert /\left\vert
A_{0}\right\vert $, which becomes $L^{2}$ in the conditional erasure protocol
per our convention in \eqref{eq:L-param-cond-erasure}.
\end{proof}

\bigskip

\noindent We now state the other implication:

\begin{theorem}
An $\left(  n,M,L,\varepsilon\right)  $ protocol for conditional erasure of a
four-system pure state $\psi_{ABER}$, as specified in
Section~\ref{sec:cond-erasure}, realizes an $(n,\left\lceil \sqrt
{M}\right\rceil ,\left\lceil \sqrt{L}\right\rceil ,4\varepsilon)$ state
redistribution protocol of $\rho_{ABE}=\operatorname{Tr}_{R}\{\psi_{ABER}\}$,
as specified in Section~\ref{sec:QSR-review}.
\end{theorem}

\begin{proof}
This follows simply by applying Uhlmann's theorem for fidelity \cite{U73} to a
conditional erasure protocol in order to realize a decoder for state
redistribution. To this end, suppose we are given a unitary quantum channel
$\mathcal{U}_{A^{n}E^{n}A^{\prime}\rightarrow A_{1}^{\prime}A_{2}^{\prime}%
\hat{E}^{n}}$ and an auxiliary catalyst state $\pi_{A^{\prime}}$, as part of a
conditional erasure protocol. Suppose further that they satisfy the negligible
disturbance condition in \eqref{eq:small-disturb-first} and the decoupled
condition in \eqref{eq:cond-decoupled}. Combining these via the triangle
inequality for $\sqrt{1-F(\rho,\sigma)}$ (similar to how we did previously in
\eqref{eq:triangle-ineq-1-F}), we find that the following condition holds%
\begin{equation}
F(\omega_{A_{1}^{\prime}B^{n}\hat{E}^{n}},\pi_{A_{1}^{\prime}}\otimes\psi
_{BE}^{\otimes n})\geq1-4\varepsilon.\label{eq:decoupled-and-neg-disturbed}%
\end{equation}
A purification of the state $\omega_{A_{1}^{\prime}B^{n}\hat{E}^{n}}$ is the
following state:%
\begin{equation}
\varsigma_{A_{1}^{\prime}A_{2}^{\prime}B^{n}\hat{E}^{n}R^{n}R^{\prime}}%
\equiv\mathcal{U}_{A^{n}E^{n}A^{\prime}\rightarrow A_{1}^{\prime}A_{2}%
^{\prime}\hat{E}^{n}}(\psi_{ABER}^{\otimes n}\otimes\Phi_{A^{\prime}R^{\prime
}}).
\end{equation}
That is, we obtain the state $\omega_{A_{1}^{\prime}B^{n}\hat{E}^{n}}$ by
tracing over the $A_{2}^{\prime}R^{n}R^{\prime}$ systems of the above state. A
purification of the state $\pi_{A_{1}^{\prime}}\otimes\psi_{BE}^{\otimes n}$
is the following state:%
\begin{equation}
\Phi_{A_{1}^{\prime}R_{1}^{\prime}}\otimes\psi_{ABER}^{\otimes n}.
\end{equation}
Thus, Uhlmann's theorem for fidelity applied to
\eqref{eq:decoupled-and-neg-disturbed}\ implies the existence of an isometric
channel $\mathcal{V}_{A_{2}^{\prime}R^{n}R^{\prime}\rightarrow R_{1}^{\prime
}A^{n}R^{n}}$ such that%
\begin{equation}
F(\mathcal{V}(\varsigma),\Phi_{A_{1}^{\prime}R_{1}^{\prime}}\otimes\psi
_{ABER}^{\otimes n})\geq1-4\varepsilon,
\end{equation}
where we have used the shorthand $\mathcal{V}(\varsigma)\equiv\mathcal{V}%
_{A_{2}^{\prime}R^{n}R^{\prime}\rightarrow R_{1}^{\prime}A^{n}R^{n}}%
(\varsigma_{A_{1}^{\prime}A_{2}^{\prime}B^{n}\hat{E}^{n}R^{n}R^{\prime}})$.
Thus, the channel $\mathcal{V}_{A_{2}^{\prime}R^{n}R^{\prime}\rightarrow
R_{1}^{\prime}A^{n}R^{n}}$ can function as a decoder for a quantum state
redistribution (QSR) protocol.

Summarizing, a purification $\Phi_{A^{\prime}R^{\prime}}$ of the catalyst
state $\pi_{A^{\prime}}$ functions as a maximally entangled resource in QSR,
the unitary channel $\mathcal{U}_{A^{n}E^{n}A^{\prime}\rightarrow
A_{1}^{\prime}A_{2}^{\prime}\hat{E}^{n}}$ functions as an encoder in QSR, the
system $A_{2}^{\prime}$ is sent over a noiseless quantum channel in QSR, the
isometric channel $\mathcal{V}_{A_{2}^{\prime}R^{n}R^{\prime}\rightarrow
R_{1}^{\prime}A^{n}R^{n}}$ functions as a decoder in QSR, and a purification
$\Phi_{A_{1}^{\prime}R_{1}^{\prime}}$ of the state $\pi_{A_{1}^{\prime}}$
functions as a maximally entangled resource shared between sender and receiver
at the end of the QSR protocol. This completes the proof.
\end{proof}

\subsection{Optimal rate region for conditional erasure}

We now define the achievable rate region for conditional erasure, which
consists of achievable rate pairs $(R_{A},R_{P})$, where $R_{A}$ is equal to
the rate of active noise and $R_{P}$ is equal to the rate of passive noise. A
rate pair $(R_{A},R_{P})$ is \textit{achievable} for conditional erasure of
$\psi_{ABER}$ if for all $\varepsilon\in\left(  0,1\right)  $, $\delta>0$, and
sufficiently large $n$, there exists an $(n,2^{n\left[  R_{A}+\delta\right]
},2^{n\left[  R_{P}+\delta\right]  },\varepsilon)$ conditional erasure
protocol. The achievable rate region of conditional erasure of $\psi_{ABER}$
is equal to the union of all rate pairs which are achievable for conditional
erasure of $\psi_{ABER}$.

Due to the equivalence between conditional erasure and state redistribution,
given in the previous section, and the results about quantum state
redistribution recalled in
\eqref{eq:redist-region-1}--\eqref{eq:redist-region-2}, we can immediately
conclude the following theorem:

\begin{theorem}
\label{thm:cond-erasure-region}The rate pair%
\begin{equation}
(I(A;B|R)_{\psi},I(A;E)_{\psi}-I(A;R)_{\psi})
\end{equation}
is achievable for conditional erasure of $\psi_{ABER}$, and the optimal
rate region is equal to%
\begin{align}
R_{A}  & \geq I(A;B|R)_{\psi},\\
R_{A}+R_{P}  & \geq2H(A|R)_{\psi}.
\end{align}

\end{theorem}

\begin{remark}
\label{rem:catalyst-sometimes-not-needed}The above theorem indicates that
sometimes a catalyst is not actually needed to complete the conditional
erasure task. In particular, if the inequality $I(A;E)_{\psi}\leq
I(A;R)_{\psi}$ holds, then the protocol generates passive noise and hence only
a vanishing, sublinear rate of passive noise is in fact needed to accomplish
the conditional erasure task. Indeed, we could double block the protocol into
$N$ blocks, each consisting of $n$ copies of $\psi_{ABER}$. For the first
block of the protocol, we could supply $\approx nI(A;E)_{\psi}$ bits of
passive noise and then the protocol would generate $\approx nI(A;R)_{\psi}$
bits of passive noise. Since the condition $I(A;E)_{\psi}\leq I(A;R)_{\psi}$
is assumed to hold, we could reinvest $\approx nI(A;E)_{\psi}$ bits of passive
noise for the second block of the protocol while generating $\approx
nI(A;R)_{\psi}$ bits of passive noise. For each block, we have an excess of
$\approx n\left[  I(A;R)_{\psi}-I(A;E)_{\psi}\right]  $ bits of passive noise
available. Repeating this procedure until the $N$th block, we find that the
rate of passive noise consumed is equal to $\approx nI(A;E)_{\psi}/nN$, since
it was only consumed in the first block, and this rate vanishes in the limit
as $n,N\rightarrow\infty$.
\end{remark}

\section{Quantum discord as einselection
cost\label{sec:einselection-cost-discord}}

Environment-induced superselection (abbrev.~\textit{einselection}) is a
process in which an interaction between a system of interest and a large
environment causes selective loss of information from the system \cite{Z03}.
The interaction with the environment has the effect of monitoring particular
observables of the system, such that only eigenstates of these observables can
persist in the system, being unaffected by the interaction. The quantum
discord was originally proposed as a measure of the decrease of correlations
after einselection is complete \cite{Z00,zurek01} and can be generalized to
include arbitrary measurements (POVMs) rather than just measurements
corresponding to system observables (see, e.g., \cite{KBCPV12,ABC16} for
reviews of discord and related measures).

To define the quantum discord, we begin with a bipartite state $\rho_{AB}$ and
a positive operator-valued measure (POVM) $\Lambda\equiv\{\Lambda_{A}^{x}\}$,
with $\Lambda_{A}^{x}\geq0$ for all $x$ and $\sum_{x}\Lambda_{A}^{x}=I_{A}$.
The (unoptimized)\ quantum discord is a measure of the loss of correlation
between $A$ and $B$ under the measurement $\Lambda$:%
\begin{equation}
D(\overline{A};B)_{\rho,\Lambda}\equiv I(A;B)_{\rho}-I(X;B)_{\zeta},
\label{eq:discord}%
\end{equation}
where%
\begin{equation}
\zeta_{XB}\equiv\sum_{x}|x\rangle\langle x|_{X}\otimes\operatorname{Tr}%
_{A}\{\Lambda_{A}^{x}\rho_{AB}\}. \label{eq:post-meas-state}%
\end{equation}

Here we continue with the main theme of this paper, namely, erasure of
correlations, and define an operational task that we call an
\textit{einselection-simulation protocol}, which is a simulation of the
einselection process via local randomizing unitaries. The starting point for
such a protocol is a bipartite state $\rho_{AB}$ and a POVM $\Lambda
\equiv\{\Lambda_{A}^{x}\}$, and the objective is to determine the minimum rate
of noise needed to apply to the $A$ system of $\rho_{AB}$, such that the
resulting state $\sigma_{AB}$\ is approximately einselected. By this, we mean that

\begin{enumerate}
\item there is a measurement corresponding to $\sigma_{AB}$, such the state
$\sigma_{AB}$ is locally recoverable after performing this measurement on
system $A$ of $\sigma_{AB}$, and

\item the corresponding post-measurement state is indistinguishable from the
post-measurement state in~\eqref{eq:post-meas-state}.
\end{enumerate}

\noindent By \cite[Proposition~21]{SW14}, the state $\sigma_{AB}$ having
negligible discord is equivalent to the condition of local recoverability of
$\sigma_{AB}$ after a measurement is performed on system$~A$.

More formally, for $n,M\in\mathbb{N}$ and $\varepsilon\in\left[  0,1\right]
$, we define an $(n,M,\varepsilon)$ einselection-simulation protocol for a
state $\rho_{AB}$ and a POVM\ $\Lambda_{A}\equiv\{\Lambda_{A}^{x}\}$ to
consist of an ensemble $\{p_{i},U_{A^{n}A^{\prime}}^{i}\}_{i=1}^{M}$ of
einselection-simulating unitaries, a catalyst state $\theta_{A^{\prime}}$, and
a measurement channel $\mathcal{M}_{A^{n}A^{\prime}\rightarrow X^{n}}$ such
that the state $\sigma_{A^{n}A^{\prime}B^{n}}$ resulting from local unitary
randomization%
\begin{equation}
\sigma_{A^{n}A^{\prime}B^{n}}\equiv\sum_{i=1}^{M}p_{i}U_{A^{n}A^{\prime}}%
^{i}\left(  \rho_{AB}^{\otimes n}\otimes\theta_{A^{\prime}}\right)  \left(
U_{A^{n}A^{\prime}}^{i}\right)  ^{\dag}%
\end{equation}
and the measurement channel $\mathcal{M}_{A^{n}A^{\prime}\rightarrow X^{n}}$
satisfy the following two requirements:

\begin{enumerate}
\item The state $\sigma_{A^{n}A^{\prime}B^{n}}$ is locally recoverable from
the classical system $X^{n}$ after the measurement channel $\mathcal{M}%
_{A^{n}A^{\prime}\rightarrow X^{n}}$ is applied, in the sense that there
exists a preparation channel $\mathcal{P}_{X^{n}\rightarrow A^{n}A^{\prime}}$
such that%
\begin{equation}
F(\sigma_{A^{n}A^{\prime}B^{n}},(\mathcal{P}\circ\mathcal{M})(\sigma
_{A^{n}A^{\prime}B^{n}}))\geq1-\varepsilon, \label{eq:ein-local-rec}%
\end{equation}
where $\mathcal{P}\equiv\mathcal{P}_{X^{n}\rightarrow A^{n}A^{\prime}}$ and
$\mathcal{M}\equiv\mathcal{M}_{A^{n}A^{\prime}\rightarrow X^{n}}$. In this
sense, we say that $\sigma_{A^{n}A^{\prime}B^{n}}$ has been approximately
einselected. In \cite{SW14}, this was described as the state $\sigma
_{A^{n}A^{\prime}B^{n}}$ being negligibly disturbed by the action of an
entanglement-breaking channel.

\item The post-measurement state $\mathcal{M}_{A^{n}A^{\prime}\rightarrow
X^{n}}$ is indistinguishable from many copies of the post-measurement state in
\eqref{eq:post-meas-state}, in the sense that%
\begin{equation}
F(\mathcal{M}_{A^{n}A^{\prime}\rightarrow X^{n}}(\sigma_{A^{n}A^{\prime}B^{n}%
}),\zeta_{XB}^{\otimes n})\geq1-\varepsilon. \label{eq:ein-faithful}%
\end{equation}
This latter condition ensures that the einselection-simulating unitaries
perform a \textit{faithful} simulation of the einselection process: they do
not destroy the correlations remaining between $X$ and $B$ after the
measurement $\Lambda_{A}$\ occurs (i.e., they only destroy the correlations in
$\rho_{AB}$ lost in the application of the measurement $\Lambda_{A}$).
\end{enumerate}

\begin{definition}
[Achievable rate]A rate $R$ of einselection simulation for a state $\rho_{AB}$
and a POVM\ $\Lambda_{A}$ is achievable if for all $\varepsilon\in(0,1)$,
$\delta>0$, and sufficiently large~$n$, there exists an $(n,2^{n\left[
R+\delta\right]  },\varepsilon)$ einselection-simulation protocol.
\end{definition}

\begin{definition}
[Einselection cost]The einselection cost $\mathcal{E}(\rho_{AB},\Lambda_{A}%
)$\ of a state $\rho_{AB}$ and a POVM\ $\Lambda_{A}$ is equal to the infimum
of all achievable rates for einselection simulation of $\rho_{AB}$ and
$\Lambda_{A}$.
\end{definition}

Our main result in this section is the following physical meaning for the
quantum discord:

\begin{theorem}
The einselection cost $\mathcal{E}(\rho_{AB},\Lambda_{A})$\ of a state
$\rho_{AB}$ and a POVM\ $\Lambda_{A}$ is equal to its quantum discord
$D(\overline{A};B)_{\rho,\Lambda}$:%
\begin{equation}
\mathcal{E}(\rho_{AB},\Lambda_{A})=D(\overline{A};B)_{\rho,\Lambda},
\end{equation}
where $D(\overline{A};B)_{\rho,\Lambda}$ is defined in \eqref{eq:discord}.
\end{theorem}

\begin{proof}
A proof of the above theorem requires two parts:\ the achievability part and
the converse. We begin with the converse, and note that it bears some
similarities to a converse given in Appendix~\ref{sec:alt-conv-proof} and the
proof of \cite[Proposition~21]{SW14}. Consider an arbitrary $(n,M,\varepsilon
)$ einselection simulation protocol for $\rho_{AB}$ and $\Lambda_{A}$, which
consists of $\{p_{i},U_{A^{n}A^{\prime}}^{i}\}_{i=1}^{M}$, $\theta_{A^{\prime
}}$, $\mathcal{M}_{A^{n}A^{\prime}\rightarrow X^{n}}$, and $\mathcal{P}%
_{X^{n}\rightarrow A^{n}A^{\prime}}$ as defined above. Let $\sigma_{\hat
{M}A^{n}A^{\prime}B^{n}}$ denote the following state:%
\begin{equation}
\sigma_{\hat{M}A^{n}A^{\prime}B^{n}}\equiv\sum_{i=1}^{M}p_{i}|i\rangle\langle
i|_{\hat{M}}\otimes U_{A^{n}A^{\prime}}^{i}(\rho_{AB}^{\otimes n}\otimes
\theta_{A^{\prime}})\left(  U_{A^{n}A^{\prime}}^{i}\right)  ^{\dag},
\end{equation}
and let $\kappa_{\hat{M}X^{n}B^{n}}$ denote the following state after the
measurement channel $\mathcal{M}_{A^{n}A^{\prime}\rightarrow X^{n}}$ acts%
\begin{equation}
\kappa_{\hat{M}X^{n}B^{n}}\equiv\mathcal{M}_{A^{n}A^{\prime}\rightarrow X^{n}%
}(\sigma_{\hat{M}A^{n}A^{\prime}B^{n}}),
\end{equation}
For such a protocol, the following chain of inequalities holds%
\begin{align}
&  \!\!\!\! nD(\overline{A};B)_{\rho,\Lambda}\nonumber\\
&  =n\left[  H(B|X)_{\zeta}-H(B|A)_{\rho}\right] \\
&  =H(B^{n}|X^{n})_{\zeta^{\otimes n}}-H(B^{n}|A^{n})_{\rho^{\otimes n}}\\
&  \leq H(B^{n}|X^{n})_{\kappa}+f(n,\varepsilon)-H(B^{n}|A^{n})_{\rho^{\otimes
n}}.
\end{align}
The first equality follows from a simple manipulation of the definition in
\eqref{eq:discord}, noting that $H(B)_{\rho}=H(B)_{\zeta}$. The second
equality follows from additivity of the conditional entropies with respect to
tensor-product states. The inequality follows from \eqref{eq:ein-faithful}
(faithfulness of the einselection simulation), the Fuchs-van-de-Graaf
inequalities in \eqref{eq:F-v-d-G-ineq}, and \cite[Lemma~2]{Winter15}, with%
\begin{equation}
f(n,\varepsilon)\equiv2n\sqrt{\varepsilon}\log\left\vert B\right\vert
+(1+\sqrt{\varepsilon})h_{2}(\sqrt{\varepsilon}/\left[  1+\sqrt{\varepsilon
}\right]  ).
\end{equation}
We now focus on bounding the two entropic terms $H(B^{n}|X^{n})_{\kappa}$ and
$-H(B^{n}|A^{n})_{\rho^{\otimes n}}$ separately. Consider that%
\begin{align}
H(B^{n}|X^{n})_{\kappa}  &  \leq H(B^{n}|A^{n}A^{\prime})_{\mathcal{P}%
(\kappa)}\\
&  \leq H(B^{n}|A^{n}A^{\prime})_{\sigma}+f(n,\varepsilon).
\end{align}
The first inequality follows because the conditional entropy does not decrease
under the action of a channel on the conditioning system, in this case the
channel being the preparation channel $\mathcal{P}_{X^{n}\rightarrow
A^{n}A^{\prime}}$. The second inequality follows from the local recoverability
condition in \eqref{eq:ein-local-rec}, the Fuchs-van-de-Graaf inequalities in
\eqref{eq:F-v-d-G-ineq}, and \cite[Lemma~2]{Winter15}. We now bound the term
$-H(B^{n}|A^{n})_{\rho^{\otimes n}}$ from above%
\begin{align}
-H(B^{n}|A^{n})_{\rho^{\otimes n}}  &  =-H(B^{n}|\hat{M}A^{n}A^{\prime
})_{\sigma_{\hat{M}}\otimes\rho^{\otimes n}\otimes\theta}\\
&  =-H(B^{n}|\hat{M}A^{n}A^{\prime})_{\sigma}\\
&  \leq-H(B^{n}|A^{n}A^{\prime})_{\sigma}+\log_{2}|\hat{M}|.
\end{align}
The first equality follows because the conditional entropy is invariant with
respect to tensoring in the product states $\sigma_{\hat{M}}\otimes
\theta_{A^{\prime}}$ to be part of the conditioning system, with $\sigma
_{\hat{M}}=\sum_{i=1}^{M}p_{i}|i\rangle\langle i|_{\hat{M}}$. The second
equality follows because the conditional entropy is invariant with respect to
the following controlled unitary acting on the systems~$\hat{M}A^{n}A^{\prime
}$ of $\sigma_{\hat{M}}\otimes\rho_{AB}^{\otimes n}\otimes\theta_{A^{\prime}}%
$:%
\begin{equation}
\sum_{i=1}^{M}|i\rangle\langle i|_{\hat{M}}\otimes U_{A^{n}A^{\prime}}^{i}.
\end{equation}
The inequality follows from a rewriting and a dimension bound for CQMI
\cite[Exercise~11.7.9]{W16} when one of the conditioned systems is classical
(in this case system $\hat{M}$):%
\begin{multline}
H(B^{n}|A^{n}A^{\prime})_{\sigma}-H(B^{n}|\hat{M}A^{n}A^{\prime})_{\sigma}\\
=I(B^{n};\hat{M}|A^{n}A^{\prime})_{\sigma}\leq\log_{2}|\hat{M}|.
\end{multline}
Putting everything together, we find the following lower bound on the rate of
an arbitrary $(n,M,\varepsilon)$ einselection simulation protocol:%
\begin{equation}
D(\overline{A};B)_{\rho,\Lambda}\leq\frac{1}{n}\log_{2}|\hat{M}|+\frac{1}%
{n}\left[  f(n,\varepsilon)+g(n,\varepsilon)\right]  .
\end{equation}
Taking the limit as $n\rightarrow\infty$ and then as $\varepsilon\rightarrow
0$, we can conclude that the quantum discord is a lower bound on the
einselection cost:%
\begin{equation}
D(\overline{A};B)_{\rho,\Lambda}\leq\mathcal{E}(\rho_{AB},\Lambda_{A}).
\label{eq:dis-ein-cost-conv}%
\end{equation}

We now turn to the achievability part, which makes use of a state
deconstruction protocol. Let $V_{AE_{0}\rightarrow XE}$ denote a unitary
extension of a measurement channel corresponding to the POVM$~\Lambda_{A}$. In
particular, we can define $V_{AE_{0}\rightarrow XE}$ as follows by its action
on a state vector $|\psi\rangle_{A}$:%
\begin{equation}
V_{AE_{0}\rightarrow XE}|\psi\rangle_{A}|0\rangle_{E_{0}}\equiv\sum_{x}\left(
\sqrt{\Lambda_{A}^{x}}|\psi\rangle_{A}\right)  _{\bar{E}}|x\rangle
_{X}|x\rangle_{\tilde{E}},
\end{equation}
where we set $E\equiv\bar{E}\tilde{E}$ and $\{|x\rangle\}_{x}$ is an
orthonormal basis. We define the isometric channel $\mathcal{V}_{A\rightarrow
XE}(\psi_{A})\equiv V(\psi_{A}\otimes|0\rangle\langle0|_{E_{0}})V^{\dag}$ and
note that tracing over system $E$ gives back the original measurement channel:%
\begin{equation}
\operatorname{Tr}_{E}\{\mathcal{V}_{A\rightarrow XE}(\psi_{A})\}=\sum
_{x}\operatorname{Tr}\{\Lambda_{A}^{x}\psi_{A}\}|x\rangle\langle x|_{X}.
\end{equation}

We now show how an $(n,M,\varepsilon)$ state deconstruction protocol for the
state $\rho_{XEB}\equiv\mathcal{V}_{A\rightarrow XE}(\rho_{AB})$ leads to an
$(n,M,9\varepsilon)$\ einselection-simulation protocol for $\rho_{AB}$ and
$\Lambda_{A}$. We consider a state deconstruction protocol in the
Landauer--Bennett erasure model. To this end, let $\mathcal{U}_{E^{n}%
X^{n}E^{\prime}\rightarrow E_{1}^{\prime}E_{2}^{\prime}X^{n}}$ be a unitary
channel, and let $\theta_{E^{\prime}}$ denote an ancilla state. Let
$\omega_{E_{1}^{\prime}E_{2}^{\prime}X^{n}B^{n}}$ denote the following state
resulting from a deconstruction operation:%
\begin{equation}
\omega_{E_{1}^{\prime}X^{n}B^{n}}\equiv\operatorname{Tr}_{E_{2}^{\prime}%
}\{\mathcal{U}_{E^{n}X^{n}E^{\prime}\rightarrow E_{1}^{\prime}E_{2}^{\prime
}X^{n}}(\rho_{XEB}^{\otimes n}\otimes\theta_{E^{\prime}})\}.
\end{equation}
The following two properties, discussed in
Section~\ref{sec:landauer-erasure-model}, hold for a state deconstruction protocol:

\begin{enumerate}
\item There exists a recovery channel $\mathcal{R}_{X^{n}\rightarrow
X^{n}E_{1}^{\prime}}$ such that system $E_{1}^{\prime}$ is locally recoverable
from $X^{n}$:%
\begin{equation}
F(\omega_{E_{1}^{\prime}X^{n}B^{n}},\mathcal{R}_{X^{n}\rightarrow X^{n}%
E_{1}^{\prime}}(\omega_{X^{n}B^{n}}))\geq1-\varepsilon.
\label{eq:rec-for-ein-2}%
\end{equation}

\item The deconstruction protocol causes neliglible disturbance to the
marginal state on systems $X^{n}B^{n}$:%
\begin{equation}
F(\omega_{X^{n}B^{n}},\rho_{XB}^{\otimes n})\geq1-\varepsilon.
\label{eq:neg-dis-for-ein}%
\end{equation}

\end{enumerate}

We now specify the components of the einselection-simulation protocol. It
consists of the following ensemble of unitaries:%
\begin{equation}
\left\{  1/M,V^{\dag\otimes n}U^{\dag}W_{E_{2}^{\prime}}^{i}UV^{\otimes
n}\right\}  _{i=1}^{M},
\end{equation}
where $U$ is the unitary operator corresponding to the unitary channel
$\mathcal{U}_{E^{n}X^{n}E^{\prime}\rightarrow E_{1}^{\prime}E_{2}^{\prime
}X^{n}}$ and $\{W_{E_{2}^{\prime}}^{i}\}_{i=1}^{M}$ is a Heisenberg--Weyl set
of unitaries for system $E_{2}^{\prime}$. The ancilla state for einselection
simulation is $\theta_{E^{\prime}}\otimes|0\rangle\langle0|_{E_{0}}$, such
that the resulting approximately einselected\ state $\sigma_{A^{n}A^{\prime
}B^{n}}$ is as follows%
\begin{equation}
\sigma_{A^{n}A^{\prime}B^{n}}\equiv V^{\dag\otimes n}U^{\dag}(\omega
_{E_{1}^{\prime}X^{n}B^{n}}\otimes\pi_{E_{2}^{\prime}})UV^{\otimes n},
\end{equation}
where we are setting system $A^{\prime}\equiv E^{\prime}E_{0}$, since the
systems $E^{\prime}E_{0}$\ will now serve as the ancilla system $A^{\prime}$
for an einselection-simulation protocol. We define the measurement channel
$\mathcal{M}_{A^{n}A^{\prime}\rightarrow X^{n}}$ as follows:%
\begin{multline}
\mathcal{M}_{A^{n}A^{\prime}\rightarrow X^{n}}(\tau_{A^{n}A^{\prime}}%
)\equiv\overline{\Delta}_{X^{n}}\circ\\
\operatorname{Tr}_{E_{1}^{\prime}E_{2}^{\prime}}\{\mathcal{U}_{E^{n}%
X^{n}E^{\prime}\rightarrow E_{1}^{\prime}E_{2}^{\prime}X^{n}}(\mathcal{V}%
_{A\rightarrow XE}^{\otimes n}(\tau_{A^{n}A^{\prime}}))\},
\end{multline}
where $\overline{\Delta}_{X^{n}}$ denotes a completely dephasing channel,
defined as%
\begin{align}
\overline{\Delta}_{X^{n}}(\xi_{X^{n}})  &  \equiv\sum_{x^{n}}|x^{n}%
\rangle\langle x^{n}|_{X^{n}}\xi_{X^{n}}|x^{n}\rangle\langle x^{n}|_{X^{n}},\\
|x^{n}\rangle &  \equiv|x_{1}\rangle_{X_{1}}\otimes\cdots\otimes|x_{n}%
\rangle_{X_{n}}.
\end{align}
We take the preparation channel $\mathcal{P}_{X^{n}\rightarrow A^{n}A^{\prime
}}$ to be%
\begin{multline}
\mathcal{P}_{X^{n}\rightarrow A^{n}A^{\prime}}(\xi_{X^{n}})\\
\equiv V^{\dag\otimes n}U^{\dag}(\mathcal{R}_{X^{n}\rightarrow X^{n}%
E_{1}^{\prime}}(\xi_{X^{n}})\otimes\pi_{E_{2}^{\prime}})UV^{\otimes n},
\end{multline}
which consists of applying the recovery channel $\mathcal{R}_{X^{n}\rightarrow
X^{n}E_{1}^{\prime}}$, appending the maximally mixed state $\pi_{E_{2}%
^{\prime}}$, inverting the deconstruction unitary $U$, and inverting the
unitary dilation $V^{\otimes n}$ of the measurement channel for $\Lambda
_{A}^{\otimes n}$.

We now demonstrate that the two conditions for einselection simulation hold.
We begin by establishing the faithfulness condition in
\eqref{eq:ein-faithful}. From definitions, we have that%
\begin{equation}
\mathcal{M}_{A^{n}A^{\prime}\rightarrow X^{n}}(\sigma_{AA^{\prime}B^{n}%
})=\overline{\Delta}_{X^{n}}(\omega_{X^{n}B^{n}}).
\label{eq:rewrite-measured-state-ein}%
\end{equation}
Furthermore, the negligible disturbance condition in
\eqref{eq:neg-dis-for-ein} and the monotonicity of fidelity with respect to
quantum channels imply that%
\begin{equation}
F(\overline{\Delta}_{X^{n}}(\omega_{X^{n}B^{n}}),\overline{\Delta}_{X^{n}%
}(\rho_{XB}^{\otimes n}))\geq1-\varepsilon.
\end{equation}
But this is equivalent to%
\begin{equation}
F(\mathcal{M}_{A^{n}A^{\prime}\rightarrow X^{n}}(\sigma_{AA^{\prime}B^{n}%
}),\rho_{XB}^{\otimes n})\geq1-\varepsilon, \label{eq:faithful-ein-ach}%
\end{equation}
by applying \eqref{eq:rewrite-measured-state-ein}\ and the fact that the
classical--quantum state $\rho_{XB}^{\otimes n}$ is invariant under the action
of the dephasing channel $\overline{\Delta}_{X^{n}}$. So this establishes the
faithfulness condition in \eqref{eq:ein-faithful}.

We now establish the local recoverability condition in
\eqref{eq:ein-local-rec}. Consider that \eqref{eq:faithful-ein-ach},
\eqref{eq:neg-dis-for-ein}, the triangle inequality for the metric $\sqrt
{1-F}$, and a rewriting imply that%
\begin{equation}
F(\mathcal{M}_{A^{n}A^{\prime}\rightarrow X^{n}}(\sigma_{AA^{\prime}B^{n}%
}),\omega_{X^{n}B^{n}})\geq1-4\varepsilon. \label{eq:close-to-deconstructed-1}%
\end{equation}
The monotonicity of fidelity with respect to quantum channels applied to
\eqref{eq:close-to-deconstructed-1}\ then implies that%
\begin{equation}
F((\mathcal{P}\circ\mathcal{M})(\sigma_{AA^{\prime}B^{n}}),\mathcal{P}%
_{X^{n}\rightarrow A^{n}A^{\prime}}(\omega_{X^{n}B^{n}}))\geq1-4\varepsilon.
\label{eq:ein-almost-done-1}%
\end{equation}
Invariance of the fidelity in \eqref{eq:rec-for-ein-2} with respect to
tensoring in $\pi_{E_{2}^{\prime}}$, applying the unitary $U^{\dag}$ followed
by $V^{\dag\otimes n}$, and applying definitions implies that%
\begin{equation}
F(\sigma_{A^{n}A^{\prime}B^{n}},\mathcal{P}_{X^{n}\rightarrow A^{n}A^{\prime}%
}(\omega_{X^{n}B^{n}}))\geq1-\varepsilon. \label{eq:ein-almost-done-2}%
\end{equation}
We can then apply the triangle inequality to \eqref{eq:ein-almost-done-1} and
\eqref{eq:ein-almost-done-2}\ with respect to the metric $\sqrt{1-F}$ and
rewrite to find that%
\begin{equation}
F(\sigma_{A^{n}A^{\prime}B^{n}},(\mathcal{P}\circ\mathcal{M})(\sigma
_{AA^{\prime}B^{n}}))\geq1-9\varepsilon.
\end{equation}
This then establishes the local recoverability condition in
\eqref{eq:ein-local-rec}. Thus, we have demonstrated that an $(n,M,\varepsilon
)$ state deconstruction protocol leads to an $(n,M,9\varepsilon)$
einselection-simulation protocol.

What remains is to show that the discord is an achievable rate for
einselection simulation. In our protocol for state deconstruction (the
particular setup considered here), an achievable rate is%
\begin{equation}
\frac{1}{n}\log_{2}\left\vert M\right\vert \approx I(E;B|X)_{\mathcal{V}%
(\rho)},
\end{equation}
which implies via the simulation argument given above that
$I(E;B|X)_{\mathcal{V}(\rho)}$ is an achievable rate for einselection
simulation. It is known from \cite{P12,SBW14} that%
\begin{equation}
D(\overline{A};B)_{\rho,\Lambda}=I(E;B|X)_{\mathcal{V}(\rho)},
\label{eq:discord-CMI}%
\end{equation}
where $\mathcal{V}(\rho)=\mathcal{V}_{A\rightarrow XE}(\rho_{AB})$, and so we
establish the inequality%
\begin{equation}
D(\overline{A};B)_{\rho,\Lambda}\geq\mathcal{E}(\rho_{AB},\Lambda_{A}),
\end{equation}
completing the proof when combined with \eqref{eq:dis-ein-cost-conv}.
\end{proof}

\begin{remark}
The operational interpretation for quantum discord given here builds upon the
previous interpretation from \cite[Section~6(c)]{W14}, given in terms of
quantum state redistribution (see \cite{MD11,CABMPW11}\ for other operational,
information-theoretic interpretations of discord). In \cite[Section~6(c)]%
{W14}, it was established via the relation in \eqref{eq:discord-CMI}\ that the
discord is equal to twice the rate of quantum communication needed in a state
redistribution protocol to transmit the environment system $E$ of
$\mathcal{V}_{A\rightarrow XE}(\rho_{AB})$ to an inaccessible environmental
system $R$, which purifies the state $\rho_{AB}$. The interpretation written
there is that discord \textquotedblleft characterizes the amount of quantum
information lost in the measurement process.\textquotedblright\ On the one
hand, we now see that the einselection-simulation protocol discussed above
perhaps gives a more natural operational interpretation of quantum discord, in
the original spirit of the discussions from \cite{Z00,zurek01}. On the other
hand, we see that at the core of the achievability proof above is the state
redistribution protocol and the method from \cite[Section~6(c)]{W14}, given
that we showed in Section~\ref{sec:simulation-via-redis}\ how state
redistribution can simulate state deconstruction.
\end{remark}

\section{Squashed entanglement}

Our main result in Theorem~\ref{thm:main} also provides an operational
interpretation of the squashed entanglement \cite{CW04}, which is an
entanglement measure satisfying many desirable properties (see \cite{LW14}%
\ and references therein). A communication-theoretic interpretation for
squashed entanglement was given in \cite{O08}, and our interpretation here
largely follows the interpretation of \cite{O08}. There, it was argued that
squashed entanglement of $\rho_{AB}$\ is equal to the fastest rate at which
Alice could send her systems to a third party possessing the best possible
quantum side information to help in decoding.

Recall that the squashed entanglement of a bipartite state $\rho_{AB}$ is
defined as%
\[
E_{\text{sq}}(A;B)_{\rho}\equiv\frac{1}{2}\inf_{\zeta_{ABE}}\left\{
I(A;B|E)_{\zeta}:\rho_{AB}=\operatorname{Tr}_{E}\{\zeta_{ABE}\}\right\}  .
\]
Due to its connection with CQMI, we thus see that the squashed entanglement is
equal to half the minimum rate of noise needed in a deconstruction operation
if Alice has available the best possible third correlated system $E$ to help
in the deconstruction task. That is, suppose that the state that Alice and Bob
begin with is $\rho_{AB}$. If there is no third system available, then the
deconstruction task reduces to decorrelating and the optimal rate of noise for
deconstructing is equal to the mutual information $I(A;B)_{\rho}$. However, if
Alice is provided with a third system $E$, such that the global state is
$\zeta_{ABE}$ with $\rho_{AB}=\operatorname{Tr}_{E}\{\zeta_{ABE}\}$, then the
rate of noise needed to achieve deconstruction is equal to $I(A;B|E)$ and
could potentially be reduced, such that fewer local randomizing unitaries are
needed in a deconstruction operation. By inspecting the formula for squashed
entanglement, we see that $E_{\text{sq}}(A;B)_{\rho}$ is equal to half the
minimum rate of noise needed in a deconstruction operation if optimal quantum
side information in $E$ is available. Also, loosely speaking, we see that the
more entangled a state is (as measured by $E_{\text{sq}}$), the more difficult
it is to deconstruct it with respect to any possible third system$~E$.

Applying the insights of \cite{LW14}, we see that squashed entanglement is
equal to half the minimum rate of noise needed to produce a state on Alice,
Bob, and Eve's systems, such that Alice's system of the resulting state is
locally recoverable from Eve's system. By \cite{LW14}, the resulting state is
thus highly extendible and furthermore arbitrarily close to a separable state
in 1-LOCC distance in the many-copy limit. In more detail, let $\omega
_{A^{\prime}B^{n}E^{n}}$ denote the state resulting from applying a state
deconstruction protocol to $\rho_{ABE}^{\otimes n}$, where $\rho_{ABE}$ is an
extension of $\rho_{AB}$. Using the argument from \cite{LW14} (repeated in
\cite{SW14}), along with the fact that $\sqrt{1-F(\rho,\sigma)}$ is a distance
measure \cite{GLN05}\ and thus obeys the triangle inequality, we find that
$F(A^{\prime};B^{n}|E^{n})_{\omega}\geq1-\varepsilon$ implies that%
\begin{equation}
\sup_{\gamma_{A^{\prime}B^{n}}\in\mathcal{E}_{k}(A^{\prime}:B^{n})}%
F(\omega_{A^{\prime}B^{n}},\gamma_{A^{\prime}B^{n}})\geq1-k^{2}\varepsilon,
\end{equation}
where $\mathcal{E}_{k}(A^{\prime}\!\!:\!\!B^{n})$ denotes the set of
$k$-extendible states, defined as the set of all states $\gamma_{A^{\prime
}B^{n}}$ such that there exists a $k$-extension $\gamma_{A_{1}^{\prime}\cdots
A_{k}^{\prime}B^{n}}$, with $\gamma_{A_{1}^{\prime}\cdots A_{k}^{\prime}B^{n}%
}$ invariant with respect to permutations of the systems $A_{1}^{\prime}\cdots
A_{k}^{\prime}$ and $\gamma_{A^{\prime}B^{n}}=\operatorname{Tr}_{A_{2}%
^{\prime}\cdots A_{k}^{\prime}B^{n}}\{\gamma_{A_{1}^{\prime}\cdots
A_{k}^{\prime}B^{n}}\}$ \cite{W89a,DPS04}. Since we can take $\varepsilon$ to
be an exponentially decreasing function of $n$ \cite{YD09}, we can take $k$
growing to infinity, say, proportional to $n^{2}$, such that $\sup
_{\gamma_{A^{\prime}B^{n}}\in\mathcal{E}_{k}(A^{\prime}:B^{n})}F(\omega
_{A^{\prime}B^{n}},\gamma_{A^{\prime}B^{n}})\rightarrow1$ as $k,n\rightarrow
\infty$. Thus, the squashed entanglement can be interpreted in terms of
$k$-extendibility as stated above.

To get the statement about 1-LOCC\ distance to separable states, we need only
apply a result from \cite{BH13}, which states that the 1-LOCC\ distance
between $k$-extendible states and separable states can be bounded from above
by a term $\propto\sqrt{(\log_{2}\left\vert A^{\prime}\right\vert )/k}$. In
our case, $\log_{2}\left\vert A^{\prime}\right\vert $ is linear in $n$, and
with $k\propto n^{2}$, the 1-LOCC distance between $k$-extendible states and
separable states vanishes in the large $n$ limit.

\section{Discussion\label{sec:conclu}}

We have provided an operational interpretation of the conditional quantum
mutual information (CQMI) $I(A;B|E)_{\rho}$ of a tripartite state $\rho_{ABE}%
$\ as the minimal rate of noise needed to apply in a deconstruction operation,
such that it has negligible disturbance of the marginal state $\rho_{BE}$
while producing a state that is locally recoverable from system $E$ alone.
Equivalently, we find that CQMI is equal to the minimal rate of noise needed
to result in a state that has vanishing normalized CQMI. The method for
showing achievability of CQMI\ in such a state deconstruction task relies upon
the quantum state redistribution protocol \cite{DY08,YD09}. We showed how the
state deconstruction protocol simplifies significantly if the system $E$ is
classical. We also considered the task of conditional erasure, in which the
goal is to apply a noisy operation to the $AE$ systems such that the $BE$
systems are negligibly disturbed and the resulting $A$ system is decoupled
from the $BE$ systems. We find again that the minimal rate of noise for
conditional erasure is equal to the CQMI\ $I(A;B|E)$. We also provided new
operational interpretations of quantum correlation measures which have
CQMI\ at their core, including quantum discord \cite{Z00,zurek01}\ and
squashed entanglement \cite{CW04}. We should also mention that our operational interpretation of CQMI seems natural in the context of the recent contribution of \cite{DHW16}, which discussed scrambling of information due to bipartite unitary interactions.

Going forward from here, we suspect that it should be possible to generalize
our results to multipartite CQMI\ quantities \cite{YHHHOS09,AHS08}. We also
think there are major obstacles to be overcome before we can determine a
satisfying one-shot generalization of these results, just as there are
obstacles in doing so for quantum state redistribution
\cite{BCT16,DHO14,ADJ14}. We would also like to know whether the CQMI\ is
generally achievable for the task of state deconstruction if no catalyst is
available. Remark~\ref{rem:catalyst-sometimes-not-needed} discusses how
a catalyst is sometimes  not actually needed for state deconstruction or
conditional erasure, but we would like to know whether this might generally be
the case.

\begin{acknowledgments}
We are indebted to Siddhartha Das, Gilad Gour, Matt Hastings, Mio Murao, Marco
Piani, Kaushik Seshadreesan, Eyuri Wakakuwa, and Andreas Winter for valuable
discussions. We also acknowledge the catalyzing role of the open problems
session at Beyond IID\ 2016, which ultimately led to the solution of the
i.i.d.~result presented here. MB acknowledges funding by the SNSF through a
fellowship, funding by the Institute for Quantum Information and Matter
(IQIM), an NSF Physics Frontiers Center (NFS Grant PHY-1125565) with support
of the Gordon and Betty Moore Foundation (GBMF-12500028), and funding support
form the ARO grant for Research on Quantum Algorithms at the IQIM
(W911NF-12-1-0521). CM acknowledges financial support from the European Research Council (ERC Grant Agreement no.~337603), the Danish Council for Independent Research (Sapere Aude) and VILLUM FONDEN via the QMATH Centre of Excellence (Grant No.~10059). MMW\ acknowledges support from the NSF\ under
Award no.~CCF-1350397.
\end{acknowledgments}

\appendix

\section{Alternative converse proofs\label{sec:alt-conv-proof}}

In this appendix, we detail two alternative converse proofs for
Theorem~\ref{thm:CQMI-lower-bnd}, which are tailored to the local unitary
randomizing model. One of the proofs bears similarities to the current proof
of Theorem~\ref{thm:CQMI-lower-bnd} and the other is similar to those
appearing in prior work \cite{GPW05,BBW15,WSM15a}. We begin with the former proof.

Let $\{p_{i},U_{A^{n}A^{\prime}E^{n}}^{i}\}_{i=1}^{M}$ denote any ensemble of
unitaries and $\theta_{A^{\prime}}$ a corresponding ancilla state realizing an
$(n,M,\varepsilon)$ state deconstruction protocol, such that the
deconstruction operation $\mathcal{N}_{A^{n}A^{\prime}E^{n}}$\ is as given in
\eqref{eq:LUR-deconstruction} and satisfies the conditions of negligible
disturbance in \eqref{eq:small-disturb-first}\ and local recoverability in \eqref{eq:local-rec}.

Let $\sigma_{\hat{M}A^{n}B^{n}E^{n}A^{\prime}}$ denote the following state:%
\begin{multline}
\sigma_{\hat{M}A^{n}B^{n}E^{n}A^{\prime}}\equiv\\
\sum_{i=1}^{M}p_{i}|i\rangle\langle i|_{\hat{M}}\otimes U_{A^{n}A^{\prime
}E^{n}}^{i}(\rho_{ABE}^{\otimes n}\otimes\theta_{A^{\prime}})(U_{A^{n}%
A^{\prime}E^{n}}^{i})^{\dag}.
\end{multline}
Then consider that%
\begin{align}
&  \!\!\!\!nI(A;B|E)_{\rho}\nonumber\\
&  =I(A^{n};B^{n}|E^{n})_{\rho^{\otimes n}}\\
&  =H(B^{n}|E^{n})_{\rho^{\otimes n}}-H(B^{n}|A^{n}E^{n})_{\rho^{\otimes n}}\\
&  \leq H(B^{n}|E^{n})_{\sigma}+f(n,\varepsilon)-H(B^{n}|A^{n}E^{n}%
)_{\rho^{\otimes n}}.
\end{align}
The inequality follows for a similar reason as given for the first inequality
in \eqref{eq:CQMI-converse-1st-block}, with $f(n,\varepsilon)$ chosen as in
\eqref{eq:f-func-continuity-AFW}. We now focus on bounding the term
$-H(B^{n}|A^{n}E^{n})_{\rho^{\otimes n}}$:%
\begin{align}
&  \!\!\!\!\!\!\!-H(B^{n}|A^{n}E^{n})_{\rho^{\otimes n}}\nonumber\\
&  =-H(B^{n}|\hat{M}A^{n}E^{n}A^{\prime})_{\sigma_{\hat{M}}\otimes
\rho^{\otimes n}\otimes\theta}\\
&  =-H(B^{n}|\hat{M}A^{n}E^{n}A^{\prime})_{\sigma}\\
&  \leq-H(B^{n}|A^{n}E^{n}A^{\prime})_{\sigma}+\log_{2}|\hat{M}|.
\end{align}
The first equality follows because we are tensoring in the product states
$\sigma_{\hat{M}}=\sum_{i=1}^{M}p_{i}|i\rangle\langle i|_{\hat{M}}$ and
$\theta_{A^{\prime}}$ for the conditioning system of the conditional entropy,
which leave it invariant. The second equality follows because the conditional
entropy is invariant under the application of the following controlled unitary
to the systems $\hat{M}A^{n}E^{n}A^{\prime}$ of $\sigma_{\hat{M}}\otimes
\rho_{ABE}^{\otimes n}\otimes\theta_{A^{\prime}}$:%
\begin{equation}
\sum_{i}|i\rangle\langle i|_{\hat{M}}\otimes U_{A^{n}A^{\prime}E^{n}}^{i}.
\end{equation}
The inequality follows from a rewriting and a dimension bound for CQMI
\cite[Exercise~11.7.9]{W16}:%
\begin{multline}
H(B^{n}|A^{n}E^{n}A^{\prime})_{\sigma}-H(B^{n}|\hat{M}A^{n}E^{n}A^{\prime
})_{\sigma}\\
=I(B^{n};\hat{M}|A^{n}E^{n}A^{\prime})_{\sigma}\leq\log_{2}|\hat{M}|.
\end{multline}
Combining these inequalities, we find that%
\begin{align}
&  \!\!\!\!nI(A;B|E)_{\rho}\nonumber\\
&  \leq H(B^{n}|E^{n})_{\sigma}-H(B^{n}|A^{n}E^{n}A^{\prime})_{\sigma
}\nonumber\\
&  \qquad+f(n,\varepsilon)+\log_{2}|\hat{M}|\\
&  =I(A^{n}A^{\prime};B^{n}|E^{n})_{\sigma}+f(n,\varepsilon)+\log_{2}|\hat
{M}|\\
&  \leq g(n,\varepsilon)+f(n,\varepsilon)+\log_{2}|\hat{M}|.
\end{align}
The inequality follows by applying the local recoverability condition in
\eqref{eq:local-rec} and because locally recoverable states have small CQMI as
reviewed in \eqref{eq:loc-rec-low-CQMI}, where we choose $g(n,\varepsilon)$ as
in \eqref{eq:loc-rec-CQMI-cont}. We can then rewrite this as%
\begin{equation}
I(A;B|E)_{\rho}\leq\frac{1}{n}\log_{2}|\hat{M}|+\frac{1}{n}\left[
g(n,\varepsilon)+f(n,\varepsilon)\right]  .
\end{equation}
By taking the limit as $n\rightarrow\infty$, then $\varepsilon\rightarrow0$,
and applying definitions, we can conclude the inequality $I(A;B|E)_{\rho}%
\leq\mathcal{D}(A;B|E)_{\rho}$.

We now detail the other proof, which is similar to those given in
\cite{GPW05,BBW15,WSM15a}. We define the following pure state:%
\begin{multline}
|\varphi\rangle_{M_{1}M_{2}A^{n}A^{\prime}E^{n}B^{n}R^{n}R^{\prime}}\equiv\\
\sum_{i}\sqrt{p_{i}}|i\rangle_{M_{1}}|i\rangle_{M_{2}}U_{A^{n}A^{\prime}E^{n}%
}^{i}|\psi\rangle_{ABER}^{\otimes n}\otimes|\theta\rangle_{A^{\prime}%
R^{\prime}},
\end{multline}
where $|\psi\rangle_{ABER}$ purifies $\rho_{ABE}$ and $|\theta\rangle
_{A^{\prime}R^{\prime}}$ purifies the ancilla $\theta_{A^{\prime}}$.\ The
state $|\varphi\rangle$ above\ is a purification of $\omega_{A^{n}A^{\prime
}B^{n}E^{n}}$ in \eqref{eq:LURed-state}\ and is helpful in our analysis.
Consider that%
\begin{align}
&  \!\!\!\!\!\!\log_{2}M\geq H(\{p_{i}\})=H(M_{1})_{\varphi}\\
&  =H(M_{2}A^{n}A^{\prime}E^{n}B^{n}R^{n}R^{\prime})_{\varphi}\\
&  \geq H(A^{n}A^{\prime}E^{n}B^{n}R^{n}R^{\prime})_{\varphi}\\
&  \geq H(A^{n}A^{\prime}E^{n}B^{n})_{\varphi}-H(R^{n}R^{\prime})_{\varphi}\\
&  =H(A^{n}A^{\prime}E^{n}B^{n})_{\omega}-H(R^{n}R^{\prime})_{\psi^{\otimes
n}\otimes\theta}\\
&  =H(A^{n}A^{\prime}E^{n}B^{n})_{\omega}-H(A^{n}B^{n}E^{n}A^{\prime}%
)_{\psi^{\otimes n}\otimes\theta}\\
&  =H(A^{n}A^{\prime}E^{n}B^{n})_{\omega}-nH(ABE)_{\psi}-H(A^{\prime}%
)_{\theta}.
\end{align}
The first inequality follows because the logarithm of the cardinality of the
probability distribution $\{p_{i}\}$ is an upper bound on its entropy
$H(\{p_{i}\})$. The first equality follows because the reduced state of
$\varphi$ on system $M_{1}$ is classical with probability distribution
$\{p_{i}\}$. The second equality follows because the entropies of the
marginals of a bipartite pure state are equal (the bipartite cut here being
between system $M_{1}$ and systems $M_{2}A^{n}A^{\prime}E^{n}B^{n}%
R^{n}R^{\prime}$). The second inequality follows the entropy cannot decrease
when adding a classical system (in this case, the $M_{2}$\ system of the
reduced state on systems $M_{2}A^{n}A^{\prime}E^{n}B^{n}R^{n}R^{\prime}$ is
classical, being decohered after a partial trace over system $M_{1}$). The
third inequality is a consequence of the Araki--Lieb triangle inequality
\cite{araki1970}, which states that $H(KL)_{\tau}\geq H(K)_{\tau}-H(L)_{\tau}$
for a bipartite state $\tau_{KL}$. The third equality follows because
$\varphi_{A^{n}A^{\prime}E^{n}B^{n}}=\omega_{A^{n}A^{\prime}E^{n}B^{n}}$ and
$\varphi_{R^{n}R^{\prime}}=\psi_{R}^{\otimes n}\otimes\theta_{R^{\prime}}$.
The fourth equality follows because the state $\psi_{ABER}^{\otimes n}%
\otimes\theta_{A^{\prime}R^{\prime}}$ is pure, so that $H(R^{n}R^{\prime
})_{\psi^{\otimes n}\otimes\theta}=H(A^{n}B^{n}E^{n}A^{\prime})_{\psi^{\otimes
n}\otimes\theta}$. The last equality follows because entropy is additive with
respect to tensor-product states. Focusing on the term $H(A^{n}A^{\prime}%
E^{n}B^{n})_{\varphi}$, we continue with%
\begin{align}
&  H(A^{n}A^{\prime}E^{n}B^{n})_{\omega}\nonumber\\
&  \geq H(A^{n}A^{\prime}E^{n})_{\omega}+H(B^{n}|E^{n})_{\omega}%
-g(n,\varepsilon)\\
&  \geq H(A^{n}A^{\prime}E^{n})_{\omega}+H(B^{n}|E^{n})_{\psi^{\otimes n}%
}-g(n,\varepsilon)-f(n,\varepsilon)\\
&  =H(A^{n}A^{\prime}\hat{E}^{n})_{\omega}+nH(B|E)_{\psi}-g(n,\varepsilon
)-f(n,\varepsilon).
\end{align}
The first inequality follows by applying the local recoverability condition in
\eqref{eq:local-rec}, the Fuchs-van-de-Graaf inequalities in
\eqref{eq:F-v-d-G-ineq}, and because locally recoverable states have small
CQMI as reviewed in \eqref{eq:loc-rec-low-CQMI}. In particular, we can take%
\begin{equation}
g(n,\varepsilon)\equiv2n\sqrt{\varepsilon}\log\left\vert B\right\vert +\left(
1+\sqrt{\varepsilon}\right)  h_{2}(\sqrt{\varepsilon}/\left[  1+\sqrt
{\varepsilon}\right]  ),
\end{equation}
and find that%
\begin{equation}
I(A^{n}A^{\prime};B^{n}|E^{n})_{\omega}\leq g(n,\varepsilon),
\end{equation}
which when rewritten is equivalent to the first inequality. The second
inequality follows from the negligible disturbance condition from
\eqref{eq:small-disturb-first}, the Fuchs-van-de-Graaf inequalities in
\eqref{eq:F-v-d-G-ineq}, and the continuity of conditional quantum entropy
\cite{AF04,Winter15}, with%
\begin{equation}
f(n,\varepsilon)=2\sqrt{\varepsilon}n\log\left\vert B\right\vert
+(1+\sqrt{\varepsilon})h_{2}(\sqrt{\varepsilon}/[1+\sqrt{\varepsilon}]).
\end{equation}
The equality holds because entropy is additive with respect to tensor-product
states. Now focusing on the term $H(A^{n}A^{\prime}E^{n})_{\omega}$, we
continue with%
\begin{align}
H(A^{n}A^{\prime}E^{n})_{\omega}  &  \geq\sum_{i}p_{i}H(A^{n}A^{\prime}%
E^{n})_{U^{i}(\psi^{\otimes n}\otimes\theta)U^{i\dag}}\\
&  =\sum_{i}p_{i}H(A^{n}E^{n}A^{\prime})_{\psi^{\otimes n}\otimes\theta}\\
&  =nH(AE)_{\psi}+H(A^{\prime})_{\theta}.
\end{align}
The first inequality follows from the concavity of quantum entropy. The first
equality follows from unitary invariance of entropy, and the last again from
additivity of entropy for tensor-product states. Putting everything together,
we find that the following bound holds for any $\left(  n,M,\varepsilon
\right)  $ state deconstruction protocol:%
\begin{equation}
\frac{1}{n}\log_{2}M+\frac{1}{n}\left[  g(n,\varepsilon)+f(n,\varepsilon
)\right]  \geq I(A;B|E)_{\rho}.
\end{equation}
By taking the limit as $n\rightarrow\infty$, then $\varepsilon\rightarrow0$,
and applying definitions, we can conclude the inequality $I(A;B|E)_{\rho}%
\leq\mathcal{D}(A;B|E)_{\rho}$.

\section{Requirement of the access to the conditioning system: Classical
example}

\label{sec:classical-example}

\begin{figure}
	\includegraphics[width=\columnwidth]{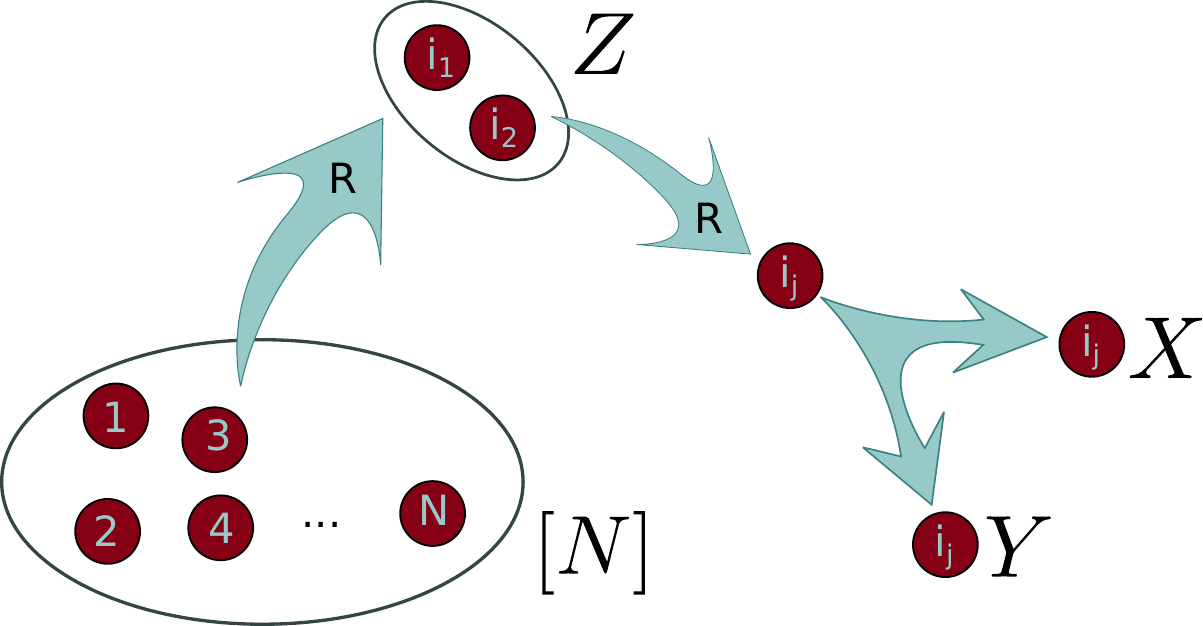}
	\caption{As an example justifying why conditional correlations cannot be erased using a noise rate equal to the CQMI without conditioning on the pivot system, consider random variables $X$, $Y$, and $Z$ constructed in the following way. $Z$ is constructed by picking a random pair of numbers $i_1\neq i_2$ from $[N]=\{1,2,...,N\}$, and $X=Y=i_b$ for a uniformly random bit $b$ that is independent of $Z$. Clearly $I(X:Y|Z)=1$. The classical analogue of a random unitary channel is a random permutation channel. When conditioning on $Z=(i,j)$, one can just randomize the pair $(i,j)$ using one bit of noise. Conversely, for a given pair $(i_1,i_2)$, the mutual information $I\left(X:Y|Z=(i_1,i_2)\right)$ is erased if and only if this pair is randomized by the random permutation channel. When operating on $X$ alone, this implies that all pairs have to be randomized, which essentially implies that all numbers $1,...,N$ have to be randomized jointly, requiring $\log_2 N$ bits of noise.}\label{fig:classical-example}
\end{figure}

The following example shows that, even for the classical analogue of state deconstruction, access to the conditioning system is necessary. Otherwise, the Markovianizing cost, as it is called in \cite{WSM15}, can be arbitrarily large compared to the conditional mutual information.
The example is informally explained in Figure~\ref{fig:classical-example}.

Let $X$ and $Y$ be random variables, each taking values from the alphabet $[n]=\{0,1,\ldots,n-1\}$, and let $Z$ be a random variable taking values in $[n]\times[n]$, such that the joint probability distribution for $(X,Y,Z)$ is as follows:
\begin{equation}
p_{XYZ}(i,j,(k,l))=\begin{cases}
\frac{1}{n(n-1)}& \text{ if }k< l,\ i=j\in\{k,l\}\\
0&\text{ else}
\end{cases}.
\end{equation}
For $k<l$, conditioned on $Z=(k,l)$, $X$ and $Y$ are maximally correlated fair coins, so that
\begin{equation}
I(X:Y|Z)=1.
\end{equation}

One classical analogue of tracing out a subsystem is applying a function $f: [n]\to[m]$ where $n/m\in\mathbb{N}$ such that $|f^{-1}(\{j\})|=n/m$ for all $j\in[m]$.

 Let $f:[n]\to[m]$ be such a function and set $X'=f(X)$. Let $X',Y,Z\sim p'$ so that
 \begin{multline}
 p'_{X'YZ}(i,j,(k,l))=
 \\
 \begin{cases}
 \frac{1}{n(n-1)}& \text{ if }k< l,\ i=f(j), j\in\{k,l\}\\
 0&\text{ else}
 \end{cases}.
 \end{multline}
 
  Let us look at the fidelity of recovery $F(X';Y|Z)$. It is easy to see that we can restrict to classical recovery channels: Let $\mathcal{R}_{Z\to ZX'}$ be an arbitrary (quantum) recovery channel, denote the measurement of the computational basis on system $X$ by $\Lambda_{X}$, etc., and let $\rho_{X'YZ}$ be the diagonal quantum state representing $p'$. The fidelity does not decrease under the application of $\Lambda_{X'}\otimes\Lambda_{Z}$ and any classical state is invariant, therefore $\mathcal{R}'_{Z\to ZX'}=(\Lambda_{X'}\otimes\Lambda_{Z})\circ \mathcal{R}_{Z\to ZX'}$ is at least as good a recovery channel as $\mathcal{R}_{Z\to ZX'}$. Now as $\Lambda_{Z}(\rho_{YZ})=\rho_{YZ}$, we can also precompose a measurement without changing the fidelity; i.e., the desired classical recovery channel that is as good as $\mathcal{R}_{Z\to ZX'}$ is $\mathcal{R}'_{Z\to ZX'}=(\Lambda_{X'}\otimes\Lambda_{Z})\circ \mathcal{R}_{Z\to ZX'}\circ\Lambda_{Z}$.
  
  Let us then take an arbitrary classical recovery channel given by a conditional probability distribution $q_{X'Z'|Z}$. The resulting recovered distribution is
  \begin{align}
	  & \hat p_{X'YZ}(i,j,\{k,l\}) \nonumber\\ 
	  &=\sum_{\substack{k'< l'\\k',l'\in [n]}}p_{YZ}(j,\{k',l'\})q_{X'Z'|Z}(i,\{k,l\}|\{k',l'\})\\
	  &=\frac{1}{n(n-1)}\sum_{\substack{k'< l'\\k',l'\in [n]}}(\delta_{k'j}+\delta_{l'j})q_{X'Z'|Z}(i,\{k,l\}|\{k',l'\})\\
	  &=\frac{1}{n(n-1)}\sum_{\substack{l'\in [n]\\ l'\neq j}}q_{X'Z'|Z}(i,\{k,l\}|\{j,l'\}).
  \end{align}
  Now we look at the fidelity with the original distribution, i.e.
  \begin{align}
	 & \sqrt{ F(p'_{X'YZ},\hat p_{X'YZ})} \nonumber \\
	 &=\sum_{\substack{k< l\\k,l\in [n]}}\sum_{i\in [m]}\sum_{j\in[n]}\sqrt{p'_{X'YZ}(i,j,\{k,l\})\hat p_{X'YZ}(i,j,\{k,l\})}\\
	  &=\frac{1}{n(n-1)}\sum_{\substack{k\neq l\\k,l\in [n]}}\sqrt{\sum_{\substack{l'\in [n]\\ l'\neq k}}q_{X'Z'|Z}(f(k),\{k,l\}|\{k,l'\})}
  \end{align}
  It is obvious that the optimal recovery channel has $q_{X'Z'|Z}(i,\{k,l\}|\{k',l'\})=0$ whenever $k,l,k',l'$ are all different or $f(k)\neq i\neq f(l)$. Let us therefore assume this is the case. Let $\lambda_{kl}=\sum_{\substack{l'\in [n]\\ l'\neq k}}q_{X'Z'|Z}(f(k),\{k,l\}|\{k,l'\})$. Then we have 
  \begin{align}\label{eq:normalize-lam}
  \sum_{\substack{k< l\\f(k)=f(l)}}\lambda_{kl}+\sum_{\substack{k\neq l\\ f(k)\neq f(l)}}\lambda_{kl}=n(n-1)/2
  \end{align}
  due to the normalization of the conditional distribution $q$. Suppose first \eqref{eq:normalize-lam} and $\lambda_{kl}\ge 0$ are the only restrictions on the possible $\lambda_{kl}$. Then the optimal choice is $$\lambda_{kl}=\frac{(n-1)}{2(n-1)-(n/m-1)},$$ i.e., constant $\lambda_{kl}$. We can now bound the fidelity of recovery
  \begin{align}
	  \sqrt{F(X';Y|Z)_{p'}}&=\max_q\sqrt{F(p',\hat p)}\\
	 &\le \max_{\lambda_{kl}\ge 0}\frac{1}{n(n-1)}\sum_{\substack{k\neq l\\k,l\in [n]}}\sqrt{\lambda_{kl}}\\
	 &= \sqrt{\frac{(n-1)}{2(n-1)-(n/m-1)}}.
  \end{align}
  Here the maxima are taken over conditional probability distributions and the positive $\lambda_{kl}$ that sum to $n(n-1)/2$, respectively. The inequality is due to the fact that by relaxing the conditions on $\lambda_{kl}$ we maximize over a larger set.
  For $F(X';Y|Z)_{p'}\ge1-\varepsilon$ this implies
  \begin{equation}
	  \log(n/m)\ge\log(n-1)+\log\left(\frac{1-2\varepsilon}{1-\varepsilon}\right) 
  \end{equation}
  In words, the required noise can be arbitrarily large compared to the conditional mutual information. A similar analysis can be done for many i.i.d.~copies of $X$, $Y$, and~$Z$.

\bibliographystyle{unsrt}
\bibliography{Ref}

\end{document}